\documentclass[conference,letterpaper]{IEEEtran}

\usepackage[utf8]{inputenc} 
\usepackage{graphicx}
\usepackage[T1]{fontenc}
\usepackage{url}
\usepackage{hyperref}
\usepackage{xcolor}
\usepackage{ifthen}
\usepackage{subfig}
\usepackage{cite}
\usepackage{amssymb}
\usepackage{amsthm}
\usepackage{float}
\usepackage[cmex10]{amsmath} 

\newcommand{\ERM}{{\sf{ERM}}}

\newcommand{\E}[1]{\mathbb E\left[#1\right]}
\newcommand{\Esub}[2]{\mathbb E_{#1}\left[#2\right]}

\DeclareMathOperator*{\argmin}{argmin}

\newtheorem{theorem}{\bf{Theorem}}
\newtheorem{definition}{\bf{Definition}}
\newtheorem{lemma}{\bf{Lemma}}
\newtheorem{corollary}{\bf{Corollary}}

\newtheorem{remark}{\bf{Remark}}
\newtheorem{example}{\bf{Example}}

\begin{document}
\title{Fast Rate Generalization Error Bounds: Variations on a Theme} 


\author{%
  \IEEEauthorblockN{Xuetong Wu$^1$, Jonathan H. Manton$^1$, Uwe Aickelin$^2$, Jingge Zhu$^{1}$}
  \IEEEauthorblockA{$^1$Department of EEE, \quad 
                    $^2$Department of CIS\\
                    University of Melbourne, Parkville, Victoria, Australia\\
                    Email: xuetongw1@student.unimelb.edu, \{jmanton, uwe.aickelin, jingge.zhu\}@unimelb.edu.au}
}


\maketitle

\begin{abstract}
A recent line of works, initiated by \cite{russo2016controlling} and \cite{xu2017information}, has shown that the generalization error of a learning algorithm can be upper bounded by information measures.
In most of the relevant works, the convergence rate of the expected generalization error is in the form of $O(\sqrt{\lambda/{n}})$ where $\lambda$ is some information-theoretic quantities such as the mutual information between the data sample and the learned hypothesis. However, such a learning rate is typically  considered to be ``slow", compared to a ``fast rate" of $O(1/n)$ in many learning scenarios. In this work, we first show that the square root does not necessarily imply a slow rate, and a fast rate $(1/n)$ result can still be obtained using this bound under appropriate assumptions. 
Furthermore, we identify the key conditions needed for the fast rate generalization error, which we call the $(\eta,c)$-central condition. Under this condition, we give information-theoretic bounds on the generalization error and excess risk, with a convergence rate of $O\left(\lambda/{n}\right)$ for specific learning algorithms such as empirical risk minimization. Finally, analytical examples are given to show the effectiveness of the bounds.
\end{abstract}


\section{Introduction} \label{sec:intro}
The generalization error of a learning algorithm lies in the core analysis of the statistical learning theory, and the estimation of which becomes remarkably crucial. Conventionally, many bounding techniques are proposed under different conditions and assumptions such as VC-dimension \cite{vapnik1999nature}, algorithmic stability \cite{bousquet_stability_2002}, PAC-Bayes \cite{mcallester1999some} and robustness \cite{xu2012robustness}. However, most bounds mentioned above are only concerned with the hypothesis or the algorithm solely. 
To fully characterize the intrinsic nature of a learning problem, it is shown in some recent works that the generalization error can be upper bounded using the information-theoretic quantities \cite{xu2017information,russo2016controlling} and the bound usually takes the following form:
\begin{align}
    \mathbb{E}_{W\mathcal{S}_n}[\mathcal{E}(W, \mathcal{S}_n)] \leq \sqrt{\frac{c I(W;\mathcal{S}_n)}{n}}, \label{eq:gen-form}
\end{align}
where the expectation is taken w.r.t. the joint distribution of $W$ and $\mathcal{S}_n$ induced by some algorithm $\mathcal{A}$. Here, $\mathcal{E}(w, \mathcal{S}_n)$ denotes the generalization error (properly defined in~(\ref{eq:gen}) in Section \ref{sec:prob}) for a given hypothesis $w$ and data sample $\mathcal{S}_n = (Z_i)_{i=1,\cdots,n}$, and $I(W;\mathcal{S}_n)$ denotes the mutual information between the hypothesis and data sample, and $c$ is some positive constant. In particular if the loss function is $\sigma$-sub-Gaussian\footnote{A random variable $X$ is $\sigma$-sub-Gaussian if $\log \mathbb{E}\left[e^{\eta(X-\mathbb{E}[X])}\right] \leq \frac{\sigma^{2} \eta^{2}}{2}$, $\forall \eta \in \mathbb{R}$.} under the distribution $P_W \otimes P_Z$, $c$ is equal to $2\sigma^2$. By introducing the mutual information, such a bound gives a data-algorithm dependent bound that can recover the previous results in terms of VC dimension~\cite{xu2017information}, algorithmic stability~\cite{raginsky2016information}, differential privacy~\cite{steinke2020reasoning} under mild conditions. Further, as pointed out by \cite{asadi_chaining_2018}, the information-theoretic upper bound could be substantially tighter than the traditional bounds if we could exploit specific properties of the learning algorithm.

However, there are mainly two issues recognized from this bound. The first problem is that, with bounded mutual information, the convergence rate is usually $O(\sqrt{1 /n})$, which is sub-optimal in some learning scenarios. The second issue is that the mutual information term can be arbitrarily large for some deterministic algorithms or some VC hypothesis class \cite{Grunwald2021pac}. The latter can be addressed by introducing ghost samples \cite{steinke2020reasoning} or using random subset methods \cite{bu2020tightening,zhou2022individually,Haghifam2020}. Only a few works are dedicated to the former problem. In this work, we develop a general framework for the fast rate bounds using the mutual information following this line of works~\cite{van2015fast,Grunwald2020,Grunwald2021pac} and the contributions are listed as follows.
\begin{itemize}
    \item We argue that the square root sign in~(\ref{eq:gen-form}) does not necessarily imply a slow rate and this bound can still achieve the fast rate. Specifically, under a proper assumption, 
    the fast rate (e.g., $O(1/n)$) is attainable if $c$ has the same order as the excess risk w.r.t. the sample size. In addition to removing the square root, we derive a novel form for the generalization error based on this variation.  
    \item Inspired by the analysis under the sub-Gaussian case, we identify the key assumptions needed for a more general fast rate learning framework, which we call $(\eta, c)$-central condition.  Compared with typical mutual information bounds, the convergence rate of the novel bound improves from $O(\sqrt{1/n})$ to $O(1/n)$ under some widely used algorithms such as empirical risk minimization (ERM) and regularized ERM. We could further extend our results to intermediate rates under the relaxed $(v,c)$-central conditions. 
    \item The fast rate results are confirmed with a few simple examples both numerically and analytically, showing the effectiveness of the proposed bounds. 
\end{itemize}

\section{Problem formulation} \label{sec:prob}
We consider the following machine learning framework starting with a set of $n$ examples that $\mathcal{S}_n = \{z_1,z_2,\cdots,z_n\}$, where each instance $z_i$ is i.i.d. drawn from some distribution $\mu$. One may wish to learn a hypothesis $w$ that exploits the properties of $\mathcal{S}_n$, with the aim of making predictions for new data correctly and efficiently. The choice of $w$ is performed within a set of member functions $\mathcal{W}$ with the possibly randomised algorithm $\mathcal{A}:\mathcal{Z}^n \rightarrow \mathcal{W}$ and we define the corresponding loss function $\ell: \mathcal{W}\times \mathcal{Z} \rightarrow \mathbb{R}$. Particularly if we consider the supervised learning problem in the following context, we can write $\mathcal{Z} = \mathcal{X} \times \mathcal{Y}$ and $z_i = (x_i, y_i)$ as a feature-label pair. Then the hypothesis $w: \mathcal{X} \rightarrow \mathcal{Y}$ can be regarded as a predictor for the input sample. We will call $(\mu, \ell, \mathcal{W}, \mathcal{A})$ a learning tuple. In a typical statistical learning problem, one may wish to minimize the \emph{expected} loss function $L_{\mu}(w) = E_{z\sim \mu}[\ell(w,z)]$. However, as the underlying distribution $\mu$ is usually unknown in practice, one may wish to learn $w$ by minimizing the empirical risk induced by the dataset $\mathcal{S}_n$, denoted as $w_{\ERM}$, such that 
\begin{equation}
w_{\ERM} = \argmin_{w\in \mathcal{W}}\frac{1}{n}\sum_{i=1}^{n}\ell(w,z_i),
\end{equation}
which will be employed as a predictor for the new data. Here we define $\hat{L}(w, \mathcal{S}_n) = \frac{1}{n}\sum_{i=1}^{n}\ell(w,z_i)$. To assess how this predictor performs on unseen samples, the generalization error is then introduced to evaluate whether a learner suffers from the over-fitting (or under-fitting). The optimal hypothesis for the true risk is defined as $w^*$ as 
\begin{equation}
w^{*} = \argmin_{w\in \mathcal{W}}E_{Z \sim \mu}[\ell (w,Z)],
\end{equation}
which is unknown in practice. For any $w \in \mathcal{W}$, we define the generalization error as
\begin{equation}
\mathcal{E}(w, \mathcal{S}_n) := \mathbb{E}_{Z\sim \mu}[\ell(w,Z)] - \frac{1}{n}\sum_{i=1}^{n}\ell(w,z_i).  \label{eq:gen}
\end{equation}
Another important metric, the excess risk, is defined as
\begin{equation}
\mathcal{R}(w) := \mathbb{E}_{Z\sim \mu}[\ell(w,Z)] - \mathbb{E}_{Z\sim \mu}[\ell(w^*,Z)].
\end{equation}
The excess risk evaluates how well a hypothesis $w$ perform with respect to $w^*$ given the data distribution $\mu$. We also define the corresponding empirical excess risk as
\begin{equation}
\hat{\mathcal{R}}(w, \mathcal{S}_n) := \frac{1}{n}\sum_{i=1}^{n}r(w,z_i),
\end{equation}
where $r(w,z) = \ell(w,z) - \ell(w^*,z)$. In the sequel, we are particularly interested in bounding the expected generalization error $\mathbb{E}_{W\mathcal{S}_n}[\mathcal{E}(W, \mathcal{S}_n)]$ and the excess risk $\mathbb{E}_{W}[\mathcal{R}(W)]$ for any $P_W$ induced by the algorithm $\mathcal{A}$.

\section{Main Results} \label{sec:main}
The recent advances show that under the sub-Gaussian assumption, the generalization error can be upper bounded using the information-theoretic quantities such as mutual information \cite{xu2017information,bu2020tightening,zhou2022individually} or conditional mutual information \cite{steinke2020reasoning}, where the bound usually takes the following form.
\begin{theorem}[\cite{bu2020tightening}]\label{thm:bu}
Suppose that $\ell({W}, {Z})$ is $\sigma$-sub-Gaussian under the distribution $P_{W} \otimes \mu$ where $P_W$ is the marginal induced the algorithm $\mathcal{A}$ and data distribution $\mu$, then
\begin{align}
    \mathbb{E}_{W\mathcal{S}_n} \left[\mathcal{E}(W, \mathcal{S}_n)\right]  \leq \frac{1}{n} \sum_{i=1}^{n} \sqrt{2 \sigma^{2} I\left(W ; Z_{i}\right)}. \label{eq:bu_result}
\end{align}
\end{theorem}
\begin{remark}
Throughout this paper, we focus on the case when  $I(W;Z_i) \sim O(1/n)$ in the sequel so the bound in (\ref{eq:bu_result})  gives a convergence rate of $O(\sqrt{1/n})$ for a constant $\sigma^2$.  This assumption holds for many learning settings such as ERM in regression problem \cite{raginsky2016information}, the Gibbs algorithm with mild assumptions \cite{raginsky2016information,aminian2021exact} and any algorithms under the general VC hypothesis classes (up to $\log n$) \cite{xu2017information, Grunwald2021pac}. 
\end{remark}
From the above result, it is usually recognized that the \emph{square root} sign prevents us from the fast rate, even in the following simple Gaussian mean estimation problem considered in \cite{bu2020tightening}. 
\begin{example}\label{sec:example}
Let $\ell(w,z_i) = (w-z_i)^2$, each sample is drawn from some Gaussian distribution, $Z_i \sim \mathcal{N}(\mu, \sigma_{N}^2)$. We consider the ERM algorithm that gives,
\begin{align*}
 W_{\ERM} = \frac{1}{n} \sum_{i=1}^{n} Z_i \sim \mathcal{N}(\mu, \frac{\sigma_{N}^2}{n}).
\end{align*}
The true generalization error can be calculated to be
\begin{align*}
   \Esub{W\mathcal{S}_n}{\mathcal{E}(W_\ERM, \mathcal{S}_n)} = \frac{2\sigma_{N}^2}{n},
\end{align*}
To evaluate the upper bound in~Theorem~\ref{thm:bu} for this example, we notice that for any $i$, $\ell(W,Z_i) \sim \frac{n+1}{n}\sigma_{N}^2 \chi_{1}^{2}$ where $\chi^2_1$ denotes the chi-squared distribution with 1 degree of freedom. Hence, the cumulant generating function can be calculated as,
\begin{align*}
\log \mathbb{E}_{P_W\otimes \mu}\left[e^{\eta(\ell(W,Z)-\mathbb{E}[\ell(W,Z)])}\right] = - \sigma_{W}^2 \eta -\frac{1}{2} \log \left(1-2\sigma_{W}^2 \eta \right),
\end{align*}
where $\eta \leq \frac{1}{2\sigma^2_W}$ and $\sigma^2_W = \frac{n+1}{n}\sigma_{N}^2$ to simplify the notation. In this case, it can be proved that,
\begin{align*}
    - \sigma_{W}^2 \eta -\frac{1}{2} \log \left(1-2\sigma_{W}^2 \eta \right) \leq \sigma_W^4\eta^2.
\end{align*}
Thus the loss is $\sqrt{2\sigma_W^4}$-sub-Gaussian under $P_{W} \otimes \mu$. We can also calculate the mutual information as 
\begin{align*}
   I(W;Z_i) = \frac{1}{2}\log\frac{n}{n-1}.
\end{align*}
Then the bound becomes
\begin{align}
    \mathbb{E}_{W\mathcal{S}_n} \left[\mathcal{E}(W, \mathcal{S}_n)\right]  \leq \frac{\sigma_{N}^2}{n} \sum_{i=1}^{n} \sqrt{2\frac{(n+1)^2}{n^2}\log\frac{n}{n-1}},
\end{align}
which will be of the order $O(\frac{1}{\sqrt{n}})$ as $n$ goes to infinity.
\end{example}
Now we show that in fact the same bound can be used to derive the correct (fast) convergence rate of $O(1/n)$, with a small yet important change on the assumption. Intuitively speaking, to achieve a fast rate bounds for both the generalization error and the excess risk in expectation, the output hypothesis of the learning algorithm must be ``good" enough compared to the optimal hypothesis $w^*$. Here we encode the notion of goodness in terms of the cumulant generating function by controlling the gap between $\ell(w,Z)$ and $\ell(w^*,Z)$. To facilitate such an idea, we make the sub-Gaussian assumption w.r.t. the excess risk and bound the generalization error as follows.
\begin{theorem}\label{thm:subgaussian}
Suppose that $r({W}, {Z})$ is $\sigma$-sub-Gaussian under distribution $P_{W} \otimes \mu$, then
\begin{align}
    \mathbb{E}_{W\mathcal{S}_n} \left[\mathcal{E}(W, \mathcal{S}_n)\right]  \leq \frac{1}{n} \sum_{i=1}^{n} \sqrt{2 \sigma^{2} I\left(W ; Z_{i}\right)}. \label{eq:our_result}
\end{align}
Furthermore, the excess risk can be bounded by,
\begin{align}
    \mathbb{E}_{W} \left[\mathcal{R}(W)\right]  \leq \mathbb{E}_{W\mathcal{S}_n} \left[\hat{\mathcal{R}}(W, \mathcal{S}_n)\right] + \frac{1}{n} \sum_{i=1}^{n} \sqrt{2 \sigma^{2} I\left(W ; Z_{i}\right)} . \label{eq:our_result_excess}
\end{align}
\end{theorem}
Now we evaluate the bound in~Theorem~\ref{thm:subgaussian} for the Gaussian example. Notice that~(\ref{eq:our_result}) is identical to (\ref{eq:bu_result}), and the only difference between Theorem~\ref{thm:bu} and Theorem~\ref{thm:subgaussian} is the assumption.
\begin{example}[Continuing from Example~\ref{sec:example}]
Consider the settings in Example~\ref{sec:example}. First we note that the expected risk minimizer $w^*$ is calculated as $\mu$. Then we have,
\begin{align*}
    r(w,z_i) = (w - z_i)^2 - (\mu - z_i)^2.
\end{align*}
The expected excess risk can be calculated as,
\begin{align*}
    \mathbb{E}_{W}[\mathcal{R}(W)] = \frac{\sigma_{N}^2}{n}.
\end{align*}
With a large $n$, we can calculate the cumulant generating function as,
\begin{align*}
\log \mathbb{E}_{P_W\otimes \mu}\left[e^{\eta(r(W,Z)-\mathbb{E}[r(W,Z)])}\right] \approx \frac{2\eta^2\sigma_{N}^4}{n},
\end{align*}
for any $\eta \in \mathbb{R}$. Hence $r(W,Z)$ is $\sqrt{\frac{4\sigma_N^4}{n}}$-sub-Gaussian under the distribution $P_{W} \otimes \mu$. Then the bounds becomes,
\begin{align*}
    \mathbb{E}_{W\mathcal{S}_n} \left[\mathcal{E}(W, \mathcal{S}_n)\right] \leq \frac{\sigma_N^2}{n} \sum_{i=1}^{n} \sqrt{\frac{4}{n}\log\frac{n}{n-1}},
\end{align*}
which is $O(1/n)$, yielding a fast rate characterization.
\end{example}
Unlike typical information-theoretic results where the bounds are based on the assumption that the loss function is $\sigma$-sub-Gaussian, we assume that the \emph{excess risk} is $\sigma$-sub-Gaussian. Even though the bound in~(\ref{eq:bu_result}) has exactly the same form as in~(\ref{eq:our_result}), the key difference is that under our assumption, $\sigma$ can depend on the sample size and will converge to $0$ as the sample size increases, while this is not the case under the previous assumption as we see in Example~\ref{sec:example}.  Moreover, the excess risk can be straightforwardly upper bounded as in~(\ref{eq:our_result_excess}). 

To make above ``fast rate" result more explicit, we provide an alternative bound based on the same subgaussian assumption. The key property of the following bound is that it does not contain the squre root.

\begin{theorem}[Fast Rate with Sub-Gaussian Condition]\label{thm:subgaussianv2}
Assume that $r(W, Z)$ is $\sigma$-subgaussian under the distribution $P_W \otimes \mu$. Then it holds that
\begin{align}
     \mathbb{E}_{W\mathcal{S}_n} \left[\mathcal{E}(W, \mathcal{S}_n)\right] \leq  \frac{1-a_\eta}{a_\eta} \mathbb{E}_{W\mathcal{S}_n}[\hat{\mathcal{R}(W,\mathcal{S}_n)}]  + \frac{1}{n\eta a_\eta} \sum_{i=1}^{n}  I\left(W ; Z_{i}\right). \label{eq:subgaussian}
\end{align}
for any $ 0 < \eta < \frac{2\mathbb{E}_{P_W \otimes \mu}[r(W,Z_i)]}{\sigma^2}$ and $a_\eta = 1-  \frac{\eta\sigma^2}{2\mathbb{E}_{P_W \otimes \mu}[r(W,Z_i)]}$. Furthermore, the expected excess risk is bounded by,
\begin{align*}
 \mathbb{E}_{W}[\mathcal{R}(W)] \leq& \frac{1}{a_\eta} \mathbb{E}_{W\mathcal{S}_n}[\hat{\mathcal{R}(W,\mathcal{S}_n)}] 
    + \frac{1}{n\eta a_\eta} \sum_{i=1}^{n}  I\left(W ; Z_{i}\right).
\end{align*}
\end{theorem}
\begin{remark}
The bound in Theorem \ref{thm:subgaussianv2} appears to provide a "fast rate" result if $I(W;Z_i)$ scales as $O(1/n)$, assumed throughout the paper. However, notice that both $\eta$ and $\alpha_\eta$ depend on the expected excess risk $\mathbb{E}_{P_W \otimes \mu}[r(W,Z)]$ and $\sigma^2$, which potentially depend on $n$ as well. Hence a more careful examination is needed. Specifically, it can be seen that if the ratio of the two quantities remains a constant independent of $n$, the fast rate result will then hold. 
\end{remark}
We continue to examine the bound in Theorem~\ref{thm:subgaussianv2} with the Gaussian mean estimation. 
\begin{example}\label{eg:subv2}
Since the expected excess risk can be calculated as $\mathbb{E}_{W}[\mathcal{R}(W)] = \frac{\sigma_N^2}{n}$, and $r(W,Z)$ is $\sqrt{\frac{4\sigma_N^4}{n}}$-sub-Gaussian, then we require that $0 <\eta < \frac{1}{2\sigma_N^2}$, which is independent of the sample size. For simplicity, we can consider the case $\eta = \frac{1}{4\sigma_N^2}$ as an example, then $a_\eta$ is calculated to be $\frac{1}{2}$. For large $n$, we have the generalization error bound,
\begin{align*}
   \frac{1-a_\eta}{a_\eta}\Esub{W\mathcal{S}_n}{\hat{\mathcal{R}}(W_\ERM,\mathcal{S}_n)} +  \frac{1}{\eta a_{\eta} n}\sum_{i=1}^{n}I(W;Z_i) \leq  \frac{3\sigma_N^2}{n},
\end{align*}
where the empirical excess risk $\Esub{W\mathcal{S}_n}{\hat{\mathcal{R}}(W_\ERM,\mathcal{S}_n)}$ is calculated as $-\frac{\sigma^2_N}{n}$ and the bound has the rate of $O(1/n)$. 
\end{example}

\subsection{Fast Rate Bound}

As discussed above, although the bound in Theorem~\ref{thm:subgaussianv2} takes the form of a "fast rate", it is still not very satisfying because it contains quantities ($\eta$ and $\alpha_\eta$) that could scale with $n$, making it hard to determine the actual convergence rate, the same as in the original bound in Theorem~\ref{thm:subgaussian}. To this end, we propose a different "fast rate" bound to alleviate this drawback. In particular, this bound does not contain extra quantities that depend on $n$. The key to this bound is the so-called the expected ($\eta,c$)-central condition (or we simply say ($\eta,c$)-central condition for short), inspired by the works  \cite{van2015fast,mehta2017fast,Grunwald2020,Grunwald2021pac}, which is the key condition leading to the fast rate. 
\begin{definition}[Expected $(\eta,c)$-Central Condition]
Let ${\eta}>0$ and $0 < c \leq 1$ be two constants. We say that $(\mu, \ell, \mathcal{W}, \mathcal{A})$ satisfies the expected $(\eta,c)$-central condition if the following inequality holds for the optimal hypothesis $w^*$:
\begin{align}
\log \mathbb{E}_{P_W\otimes \mu}& \left[e^{-{\eta}\left(\ell(W,Z)-\ell(w^*,Z)\right)}\right]  \leq  \nonumber \\
& -c\eta  \mathbb{E}_{P_W\otimes \mu}\left[\ell(W,Z) - \ell(w^*,Z)\right]. \label{eq:eta_c} 
\end{align} 
\end{definition}
Compared to the conventional $\eta$-central condition \cite[Def. 3.1]{van2015fast} by setting $c =0$ in (\ref{eq:eta_c}) as
\begin{align}
\log \mathbb{E}_{P_W\otimes \mu}& \left[e^{-{\eta}\left(\ell(w,Z)-\ell(w^*,Z)\right)}\right]  \leq 0, \label{eq:eta} 
\end{align}
the RHS of~(\ref{eq:eta_c}) is negative and has a tighter control than (\ref{eq:eta}) of the tail behaviour for some $c> 0$. Roughly speaking, this condition ensures that for all $\eta' \leq \eta$, the probability that $w$ outperforms $w^*$ by more than $L$ is exponentially small in $(c+1)L$. 
We point out that such a condition is indeed the key assumption for generalizing the result of Theorem~\ref{thm:subgaussianv2}, which also coincides with some well-known conditions that lead to a fast rate. We firstly show that the Bernstein condition\cite{bartlett2006empirical,bartlett2006convexity, hanneke2016refined,mhammedi2019pac} implies the $(\eta,c)$-central condition for certain $\eta$ and $c$ in the following corollary. 
\begin{corollary}\label{coro:berstein}
Let $\beta \in [0,1]$ and $B \geq 1$. For a learning tuple $(\mu, \ell, \mathcal{W}, \mathcal{A})$,  we say that the \textbf{Bernstein condition} holds if the following inequality holds for the optimal hypothesis $w^*$:
\begin{align*}
       {\mathbb{E}}_{P_W \otimes \mu}&\left[\left(\ell\left(W, Z^{\prime}\right)-\ell\left(w^{*} , Z^{\prime}\right)\right)^{2}\right] \\ 
       & \leq B \left({\mathbb{E}}_{P_W \otimes \mu}\left[\ell\left(W, Z^{\prime}\right)- \ell\left(w^{*} ; Z^{\prime}\right)\right]\right)^{\beta}.
\end{align*}
Then, if $\beta = 1$ and $r(w,z_i)$ is bounded by $-b$ with some $b > 0$ for all $w$ and $z_i$, the learning tuple also satisfies $(\min(\frac{1}{b}, \frac{1}{2B(e-2)}), \frac{1}{2})$-central condition.
\end{corollary}
The Bernstein condition is usually recognized as a characterization of ``easiness" of the learning problem under various $\beta$ where $\beta = 1$ corresponds to the ``easiest" learning case. For bounded loss functions, the Bernstein condition will automatically hold with $\beta = 0$.  The standard Bernstein condition requires that the inequality holds for any $w\in \mathcal{W}$, which is usually difficult to satisfy even in some trivial examples as we will see in Example~\ref{sec:example}. Different from the standard setting, we only require that the learned (randomised) hypothesis $W$ satisfy the inequality in expectation. This is a weaker but more natural condition in the sense that we do not expect any $w \in \mathcal{W}$ will work but hope that the algorithm outputs the hypothesis that performs well in average. 

The second condition is the central condition with the witness condition \cite{van2015fast,Grunwald2020}, which also implies the $(\eta,c)$-central condition. We say $(\mu, \ell, \mathcal{W}, \mathcal{A})$ satisfies the $\eta$-central condition \cite{van2015fast,Grunwald2020} if for the optimal hypothesis $w^*$, the following inequality holds,
\begin{align*}
\mathbb{E}_{P_W\otimes \mu}\left[e^{-{\eta}\left(\ell(W,Z)-\ell(w^*,Z)\right)}\right] \leq 1. 
\end{align*}
We also say the learning tuple  $(\mu, \ell, \mathcal{W}, \mathcal{A})$ satisfies the  $(u, c)$-witness condition \cite{Grunwald2020} if for constants $u > 0$ and $c \in (0,1]$, the following inequality holds.
\begin{align*}
     \mathbb{E}_{P_W\otimes \mu}& [\left(\ell(W,Z)-\ell({w^{*}},Z) \right) \cdot   \mathbf{1}_{\left\{\ell(W,Z)-\ell({w^{*}},Z) \leq u \right\}}] \\
     &\geq c \mathbb{E}_{P_W\otimes \mu}\left[\ell(W,Z) -\ell({w^{*}},Z) \right],
\end{align*}
where $ \mathbf{1}_{\{\cdot\}}$ denotes the indicator function. Then we have the following corollary.
\begin{corollary}\label{coro:central}
If the learning tuple satisfies both $\eta$-central condition and $(u,c)$-witness condition, then the learning tuple also satisfies the $(\eta', \frac{c-\frac{c\eta'}{\eta}}{\eta' u +1})$-central condition for any $0 < \eta' < \eta$.
\end{corollary}
The standard $\eta$-central condition is a key condition for proving the fast rate \cite{van2015fast,mehta2017fast,Grunwald2020}. Some examples are exponential concave loss functions (including log-loss) with $\eta = 1$ (see \cite{mehta2017fast,zhu2020semi} for examples) and bounded loss functions with Massart noise condition with different $\eta$ \cite{van2015fast}.  Again, different from the standard central condition, we only require that it holds in expectation w.r.t. the distribution induced by the algorithm $\mathcal{A}$. The witness condition~\cite[Def. 12]{Grunwald2020} is imposed to rule out situations in which learnability simply cannot hold. The intuitive interpretation of this condition is that we exclude bad hypothesis $w$ with negligible probability (but still can contribute to the expected loss), which we will never witness empirically. 
With the definitions in place, we derive the fast rate bounds under the $(\eta, c)$-central condition as follows. 
\begin{theorem}[Fast Rate with $(\eta, c)$-central condition]\label{thm:eta-c}
Assume the learning tuple $(\mu, \ell, \mathcal{W}, \mathcal{A})$ satisfies the  $\left(\eta, c \right)$-central condition for some constants $\eta > 0$ and $0 < c \leq 1$. Then, for all $\eta' \in\left(0, \eta \right]$, it holds that,
\begin{align*}
     \mathbb{E}_{W\mathcal{S}_n}[\mathcal{E}(W,\mathcal{S}_n)] \leq & \frac{1-c}{c} \mathbb{E}_{P_{W\mathcal{S}_n}}[\hat{\mathcal{R}}\left(W, \mathcal{S}_{n} \right)] + \frac{1}{c\eta' n} \sum_{i=1}^{n} I(W;Z_i).
\end{align*}
\noindent Furthermore, the excess risk is bounded by,
 \begin{align*}
     \mathbb{E}_{W}[\mathcal{R}(W)] \leq & \frac{1}{c} \mathbb{E}_{P_{W\mathcal{S}_n}}[\hat{\mathcal{R}}\left(W, \mathcal{S}_{n} \right)]  + \frac{1}{c\eta' n} \sum_{i=1}^{n} I(W;Z_i).
 \end{align*}
\end{theorem}
\noindent Such a bound has similar form with \cite[Eq.~(3)]{Grunwald2021pac} which consists of the empirical excess risk and mutual information terms and the first term is negative for some algorithms such as ERM. Notice that different from the bound in Theorem~\ref{thm:subgaussianv2}, the bound in Theorem~\ref{thm:eta-c} contains constants $c$ and $\eta'$ that do not depend on the sample size $n$. By absorbing the necessary dependence on $n$ in the definition of the ($\eta,c$)-central condition. Now it is  instructive to compare the different assumptions used in the above bounds. A summary of the key technical conditions is presented in Table~\ref{tab:tech1} for easier comparisons, while some comments are provided in Section~\ref{sec:related}.  In this case, the convergence rate will  depend on the mutual information $I(W;Z_i)$, which can achieve fast rate of $O({1}/{n})$ for appropriate learning problems and algorithms \cite{aminian2021exact,steinke2020reasoning,Grunwald2021pac}. In the following we analytically examine our bounds in Gaussian mean estimation, and we also empirically verify our bounds with a logistic regression problem in Appendix~\ref{sec:logistic}. 
\begin{example}
We can examine whether the Gaussian mean estimation satisfies the $(\eta, c)$-central condition. It can be checked that for sufficiently large $n$, 
\begin{align*}
\log \mathbb{E}_{P_W\otimes Z}\left[e^{-\eta r(W,Z)}\right] \approx \frac{2 \eta^2\sigma_N^4 - \eta \sigma_N^2}{n} \leq -c\eta \frac{\sigma_N^2}{n}.
\end{align*}
From the above inequality, this learning problem satisfy the $(\eta, c)$-central condition for any $0 < \eta < \frac{1}{2\sigma_N^2}$ and any $c \leq 1- 2\eta\sigma_N^2$, which is independent of the sample size and thus does not affect the convergence rate. Similarly, take $\eta = \frac{1}{4\sigma_N^2}$ and $c = \frac{1}{2}$, the bound becomes
\begin{align*}
    \frac{1-c}{c} \mathbb{E}_{P_{W\mathcal{S}_n}}[\hat{\mathcal{R}}\left(W, \mathcal{S}_{n} \right)] + \frac{1}{c\eta' n} \sum_{i=1}^{n} I(W;Z_i) = \frac{3\sigma_N^2}{n},
\end{align*}
which coincides with the bound in~Example~\ref{eg:subv2} and we can arrive at the fast rate since $I(W;Z_i) \sim O(1/n)$. 
\end{example}
\noindent Moreover, the learning bound in Theorem~\ref{thm:eta-c} can be applied to the regularized ERM algorithm as:
\begin{align*}
    w_{\sf{RERM}} = \argmin_{w \in \mathcal{W}} \hat{L}(w,\mathcal{S}_n) + \frac{\lambda}{n}g(w),
\end{align*}
where $g : \mathcal{W} \rightarrow \mathbb{R}$ denotes the regularizer function and $\lambda$ is some coefficient. We define $\hat{\mathcal{R}}_{\textup{reg}}(w,\mathcal{S}_n) = \hat{\mathcal{R}}(w,\mathcal{S}_n) + \frac{\lambda}{n}(g(w) - g(w^*))$, then we have  the following lemma.
\begin{lemma}\label{lemma:rerm}
We assume conditions in Theorem~\ref{thm:eta-c} hold and also assume $|g(w_1) - g(w_2)| \leq B$ for any $w_1$ and $w_2$ in $\mathcal{W}$ with some $B >0$. Then for $W_{\sf{RERM}}$:
\begin{align*}
     \mathbb{E}_{W}[\mathcal{R}(W_{\sf{RERM}})] \leq & \frac{1}{c} \mathbb{E}_{P_{W\mathcal{S}_n}}[\hat{\mathcal{R}}_{\textup{reg}}\left(W_{\sf{RERM}}, \mathcal{S}_{n} \right)]  +\frac{\lambda B}{cn} \\
     &+ \frac{1}{c\eta' n} \sum_{i=1}^{n} I(W_{\sf{RERM}};Z_i). 
\end{align*}
\end{lemma}
\noindent As $\hat{\mathcal{R}}_{\textup{reg}}(w,\mathcal{S}_n)$ will be negative for $w_{\sf{RERM}}$, the regularized ERM algorithm can lead to the fast rate if $I(W_{\sf{RERM}};Z_i) \sim O(1/n)$, which coincides with results in \cite{koren2015fast}.

\subsection{Intermediate Rate Bound}
From Theorem~\ref{thm:eta-c} we can achieve the fast rate if the mutual information between the hypothesis and data example is converging with $O(1/n)$. To further relax the $(\eta,c)$-central condition, we can also derive the intermediate rate with the order of $O(n^{-\alpha})$ for $\alpha \in [\frac{1}{2}, 1]$. Similar to the $v$-central condition, which is a weaker condition of the $\eta$-central condition \cite{van2015fast,Grunwald2020}, we propose the $(v,c)$-central condition first and derive the intermediate rate results in Theorem~\ref{lemma:intermediate}.
\begin{definition}[$(v,c)$-Central Condition]\label{def:weaker-eta-c}
We say that $(\mu, \ell, \mathcal{W}, \mathcal{A})$ satisfies the $(\eta,c)$-central condition up to some $\epsilon > 0$ if the following inequality holds for the optimal hypothesis $w^*$:
\begin{align}
\log \mathbb{E}_{P_W\otimes \mu} & \left[e^{-{\eta}\left(\ell(W,Z)-\ell(w^*,Z)\right)}\right]  \leq \nonumber \\
&-c\eta  \mathbb{E}_{P_W\otimes \mu}\left[\ell(W,Z) - \ell(w^*,Z)\right] + \eta \epsilon. \label{eq:v-central} 
\end{align}
Let $v:[0, \infty) \rightarrow[0, \infty)$ is a bounded and non-decreasing function satisfying $v(\epsilon)>0$ for all $\epsilon > 0$. We say that $(\mu, \ell, \mathcal{W}, \mathcal{A})$  satisfies the $(v,c)$-central condition if for all $\epsilon \geq 0$ such that~(\ref{eq:v-central}) is satisfied with $\eta = v(\epsilon)$.
\end{definition}
\begin{theorem}\label{lemma:intermediate}
Assume the learning tuple $(\mu, \ell, \mathcal{W}, \mathcal{A})$ satisfies the  $\left(v, c\right)$-central condition up to $\epsilon$ for some function $v$ as defined in~Def. \ref{def:weaker-eta-c} and $0 < c < 1$. Then it holds that for any $\epsilon \geq 0$ and any $0< \eta' \leq v(\epsilon)$,
\begin{align*}
     \mathbb{E}_{W\mathcal{S}_n}[\mathcal{E}(W,\mathcal{S}_n)] \leq & \frac{1-c}{c} \mathbb{E}_{P_{W\mathcal{S}_n}}[\hat{\mathcal{R}}\left(W, \mathcal{S}_{n} \right)] \\
     &+ \frac{1}{n} \sum_{i=1}^{n} \left( \frac{1}{\eta' c}I(W;Z_i) + \frac{\epsilon}{c}\right).
 \end{align*}
\end{theorem}
\noindent In particular, if $v(\epsilon) \asymp \epsilon^{1-\beta}$ for some $\beta \in [0,1]$, then the generalization error is bounded by,
\begin{align*}
     \mathbb{E}_{W\mathcal{S}_n}[\mathcal{E}(W,\mathcal{S}_n)] \leq & \frac{1-c}{c} \mathbb{E}_{P_{W\mathcal{S}_n}}[\hat{\mathcal{R}}\left(W, \mathcal{S}_{n} \right)]  + \frac{2}{nc}\sum_{i=1}^{n} I(W;Z_i)^{\frac{1}{2-\beta}}.
 \end{align*}
Thus, the expected generalization is found to have an order of $I(W;Z_i)^{\frac{1}{2-\beta}}$, which corresponds to the results under Bernstein's condition \cite{hanneke2016refined,mhammedi2019pac,Grunwald2021pac}. 

\begin{small}
\begin{table}[h]
    \centering
    \caption{Technical Conditions Comparisons}\label{tab:tech1}
    \begin{tabular}{|c|c|}
    \hline
     Condition      &  Key Inequality   \\
     \hline
     $(\eta,c)$-Central Condition & $\log \mathbb{E}\left[e^{-\eta r(W,Z)}\right] \leq  -c \eta \mathbb{E}[r(W,Z)]$ \\
     \hline 
     Bernstein Condition with $\beta = 1$   &  $\log  \mathbb{E}\left[e^{-\eta r(W,Z)}\right] \leq  -\frac{1}{2} \eta \mathbb{E}[r(W,Z)]$ \\
     \hline 
     Central + Witness Condition    &  $\log  \mathbb{E}\left[e^{-\eta r(W,Z)}\right] \leq  -\frac{1}{c_u}\eta\mathbb{E}[r(W,Z)]$   \\
    \hline 
     Central Condition Only &  $\log  \mathbb{E}\left[e^{-\eta r(W,Z)}\right]  \leq 0$      \\
     \hline
     SubGaussian Condition  & $\log \mathbb{E}\left[ e^{ -\eta r(W,Z)} \right] \leq -\eta \mathbb{E}[ r(W,Z)] +\frac{\eta^2 \sigma^2}{2}$    \\
     \hline 
    \end{tabular}
\end{table}
\end{small}

\subsection{Connection to other works} \label{sec:related}
Fast rate conditions are widely investigated under different learning frameworks and conditions \cite{van2015fast, mehta2017fast, koren2015fast, mhammedi2019pac, Grunwald2020, zhu2020semi, Grunwald2021pac}. We propose the $(\eta,c)$-central condition, a stronger condition than $\eta$-central condition, that can lead to the fast rate for the generalization error in expectation, which also coincides with many existing works such as \cite{Grunwald2020} and \cite{Grunwald2021pac} for certain choices of $c$ and $\eta$. In particular, with bounded loss, $\beta = 1$ in the Bernstein condition is equivalent to the central condition with the witness condition for fast rate, from which $(\eta,c)$-central condition follows. As an example of unbounded loss functions, the log-loss will satisfy the central and witness conditions under well-specified model \cite{wong1995probability,Grunwald2020}, which also consequently implies the $(\eta,c)$-central condition. As suggested by Theorem~\ref{thm:subgaussianv2}, the sub-Gaussian condition can also satisfy the $(\eta,c)$-central condition if it satisfies that $\eta\sigma^2 = a\mathbb{E}[r(W,Z)]$ for some constant $a\in (0,2)$. 

As the most relevant work, our bound is similar to that found in \cite{Grunwald2021pac} which applies conditional mutual information \cite{steinke2020reasoning}, but their results are derived under the PAC-Bayes framework and rely on prior knowledge. Our result applies to general algorithms with mutual information and our assumptions are weaker since we only require the proposed conditions hold in expectation w.r.t. $P_W$, instead of for all $w \in\mathcal{W}$. Our results also have the benefit of allowing the convergence factors to be further improved by using different metrics and data-processing techniques, see \cite{jiao2017dependence, hafez2020conditioning, zhou2022individually} for examples.

\bibliographystyle{IEEEtran}
\bibliography{reference}

\begin{thebibliography}{10}
\providecommand{\url}[1]{#1}
\csname url@samestyle\endcsname
\providecommand{\newblock}{\relax}
\providecommand{\bibinfo}[2]{#2}
\providecommand{\BIBentrySTDinterwordspacing}{\spaceskip=0pt\relax}
\providecommand{\BIBentryALTinterwordstretchfactor}{4}
\providecommand{\BIBentryALTinterwordspacing}{\spaceskip=\fontdimen2\font plus
\BIBentryALTinterwordstretchfactor\fontdimen3\font minus
  \fontdimen4\font\relax}
\providecommand{\BIBforeignlanguage}[2]{{%
\expandafter\ifx\csname l@#1\endcsname\relax
\typeout{** WARNING: IEEEtran.bst: No hyphenation pattern has been}%
\typeout{** loaded for the language `#1'. Using the pattern for}%
\typeout{** the default language instead.}%
\else
\language=\csname l@#1\endcsname
\fi
#2}}
\providecommand{\BIBdecl}{\relax}
\BIBdecl

\bibitem{russo2016controlling}
D.~Russo and J.~Zou, ``Controlling bias in adaptive data analysis using
  information theory,'' in \emph{Artificial Intelligence and Statistics}.\hskip
  1em plus 0.5em minus 0.4em\relax PMLR, 2016, pp. 1232--1240.

\bibitem{xu2017information}
A.~Xu and M.~Raginsky, ``Information-theoretic analysis of generalization
  capability of learning algorithms,'' in \emph{Proceedings of the 31st
  International Conference on Neural Information Processing Systems}, 2017, pp.
  2521--2530.

\bibitem{vapnik1999nature}
V.~Vapnik, \emph{The nature of statistical learning theory}.\hskip 1em plus
  0.5em minus 0.4em\relax Springer science \& business media, 1999.

\bibitem{bousquet_stability_2002}
O.~Bousquet and A.~Elisseeff, ``Stability and {Generalization},'' \emph{Journal
  of Machine Learning Research}, vol.~2, no. Mar, pp. 499--526, 2002.

\bibitem{mcallester1999some}
D.~A. McAllester, ``Some pac-bayesian theorems,'' \emph{Machine Learning},
  vol.~37, no.~3, pp. 355--363, 1999.

\bibitem{xu2012robustness}
H.~Xu and S.~Mannor, ``Robustness and generalization,'' \emph{Machine
  learning}, vol.~86, no.~3, pp. 391--423, 2012.

\bibitem{raginsky2016information}
M.~Raginsky, A.~Rakhlin, M.~Tsao, Y.~Wu, and A.~Xu, ``Information-theoretic
  analysis of stability and bias of learning algorithms,'' in \emph{2016 IEEE
  Information Theory Workshop (ITW)}.\hskip 1em plus 0.5em minus 0.4em\relax
  IEEE, 2016, pp. 26--30.

\bibitem{steinke2020reasoning}
T.~Steinke and L.~Zakynthinou, ``Reasoning about generalization via conditional
  mutual information,'' in \emph{Conference on Learning Theory}.\hskip 1em plus
  0.5em minus 0.4em\relax PMLR, 2020, pp. 3437--3452.

\bibitem{asadi_chaining_2018}
A.~Asadi, E.~Abbe, and S.~Verdu, ``Chaining {Mutual} {Information} and
  {Tightening} {Generalization} {Bounds},'' in \emph{Advances in {Neural}
  {Information} {Processing} {Systems} 31}, S.~Bengio, H.~Wallach,
  H.~Larochelle, K.~Grauman, N.~Cesa-Bianchi, and R.~Garnett, Eds.\hskip 1em
  plus 0.5em minus 0.4em\relax Curran Associates, Inc., 2018, pp. 7234--7243.

\bibitem{Grunwald2021pac}
P.~Gr{\"u}nwald, T.~Steinke, and L.~Zakynthinou, ``Pac-bayes, mac-bayes and
  conditional mutual information: Fast rate bounds that handle general vc
  classes,'' \emph{arXiv preprint arXiv:2106.09683}, 2021.

\bibitem{bu2020tightening}
Y.~Bu, S.~Zou, and V.~V. Veeravalli, ``Tightening mutual information-based
  bounds on generalization error,'' \emph{IEEE Journal on Selected Areas in
  Information Theory}, vol.~1, no.~1, pp. 121--130, 2020.

\bibitem{zhou2022individually}
R.~Zhou, C.~Tian, and T.~Liu, ``Individually conditional individual mutual
  information bound on generalization error,'' \emph{IEEE Transactions on
  Information Theory}, 2022.

\bibitem{Haghifam2020}
\BIBentryALTinterwordspacing
M.~Haghifam, J.~Negrea, A.~Khisti, D.~M. Roy, and G.~K. Dziugaite, ``{Sharpened
  Generalization Bounds based on Conditional Mutual Information and an
  Application to Noisy, Iterative Algorithms},'' no. NeurIPS, 2020. [Online].
  Available: \url{http://arxiv.org/abs/2004.12983}
\BIBentrySTDinterwordspacing

\bibitem{van2015fast}
T.~Van~Erven, P.~Grunwald, N.~A. Mehta, M.~Reid, R.~Williamson \emph{et~al.},
  ``Fast rates in statistical and online learning,'' 2015.

\bibitem{Grunwald2020}
P.~D. Gr{\"u}nwald and N.~A. Mehta, ``Fast rates for general unbounded loss
  functions: From erm to generalized bayes.'' \emph{J. Mach. Learn. Res.},
  vol.~21, pp. 56--1, 2020.

\bibitem{aminian2021exact}
G.~Aminian, Y.~Bu, L.~Toni, M.~Rodrigues, and G.~Wornell, ``An exact
  characterization of the generalization error for the gibbs algorithm,''
  \emph{Advances in Neural Information Processing Systems}, vol.~34, 2021.

\bibitem{mehta2017fast}
N.~Mehta, ``Fast rates with high probability in exp-concave statistical
  learning,'' in \emph{Artificial Intelligence and Statistics}.\hskip 1em plus
  0.5em minus 0.4em\relax PMLR, 2017, pp. 1085--1093.

\bibitem{bartlett2006empirical}
P.~L. Bartlett and S.~Mendelson, ``Empirical minimization,'' \emph{Probability
  theory and related fields}, vol. 135, no.~3, pp. 311--334, 2006.

\bibitem{bartlett2006convexity}
P.~L. Bartlett, M.~I. Jordan, and J.~D. McAuliffe, ``Convexity, classification,
  and risk bounds,'' \emph{Journal of the American Statistical Association},
  vol. 101, no. 473, pp. 138--156, 2006.

\bibitem{hanneke2016refined}
S.~Hanneke, ``Refined error bounds for several learning algorithms,'' \emph{The
  Journal of Machine Learning Research}, vol.~17, no.~1, pp. 4667--4721, 2016.

\bibitem{mhammedi2019pac}
Z.~Mhammedi, P.~D. Grunwald, and B.~Guedj, ``Pac-bayes un-expected bernstein
  inequality,'' \emph{arXiv preprint arXiv:1905.13367}, 2019.

\bibitem{zhu2020semi}
J.~Zhu, ``Semi-supervised learning: the case when unlabeled data is equally
  useful,'' in \emph{Conference on Uncertainty in Artificial
  Intelligence}.\hskip 1em plus 0.5em minus 0.4em\relax PMLR, 2020, pp.
  709--718.

\bibitem{koren2015fast}
T.~Koren and K.~Levy, ``Fast rates for exp-concave empirical risk
  minimization,'' \emph{Advances in Neural Information Processing Systems},
  vol.~28, 2015.

\bibitem{wong1995probability}
W.~H. Wong and X.~Shen, ``Probability inequalities for likelihood ratios and
  convergence rates of sieve mles,'' \emph{The Annals of Statistics}, pp.
  339--362, 1995.

\bibitem{jiao2017dependence}
J.~Jiao, Y.~Han, and T.~Weissman, ``Dependence measures bounding the
  exploration bias for general measurements,'' in \emph{2017 IEEE International
  Symposium on Information Theory (ISIT)}.\hskip 1em plus 0.5em minus
  0.4em\relax IEEE, 2017, pp. 1475--1479.

\bibitem{hafez2020conditioning}
H.~Hafez-Kolahi, Z.~Golgooni, S.~Kasaei, and M.~Soleymani, ``Conditioning and
  processing: Techniques to improve information-theoretic generalization
  bounds,'' \emph{Advances in Neural Information Processing Systems}, vol.~33,
  pp. 16\,457--16\,467, 2020.

\bibitem{gao2017estimating}
W.~Gao, S.~Kannan, S.~Oh, and P.~Viswanath, ``Estimating mutual information for
  discrete-continuous mixtures,'' \emph{Advances in neural information
  processing systems}, vol.~30, 2017.

\bibitem{moddemeijer1989estimation}
R.~Moddemeijer, ``On estimation of entropy and mutual information of continuous
  distributions,'' \emph{Signal processing}, vol.~16, no.~3, pp. 233--248,
  1989.

\bibitem{kraskov2004estimating}
A.~Kraskov, H.~St{\"o}gbauer, and P.~Grassberger, ``Estimating mutual
  information,'' \emph{Physical review E}, vol.~69, no.~6, p. 066138, 2004.

\bibitem{cesa2006prediction}
N.~Cesa-Bianchi and G.~Lugosi, \emph{Prediction, learning, and games}.\hskip
  1em plus 0.5em minus 0.4em\relax Cambridge university press, 2006.

\bibitem{boucheron_concentration_2013}
S.~Boucheron, G.~Lugosi, and P.~Massart,
  \emph{\BIBforeignlanguage{en}{Concentration {Inequalities}: {A}
  {Nonasymptotic} {Theory} of {Independence}}}.\hskip 1em plus 0.5em minus
  0.4em\relax OUP Oxford, Feb. 2013.

\end{thebibliography}

\newpage

\onecolumn

\appendix

\subsection{Logistic Regression Example}\label{sec:logistic}
We apply our bound in a typical classification problem. Consider a logistic regression problem in a 2-dimensional space. For each $w \in \mathbb{R}^2$ and $z_i = (x_i,y_i) \in \mathbb{R}^{2} \times \{0,1\}$, the loss function is given by
\begin{align*}
    \ell(w,z_i) := -(y_i\log (\sigma(w^Tx_i)) + (1-y_i)\log (1 - \sigma(w^Tx_i)))
\end{align*}
where $\sigma(x) = \frac{1}{1+e^{-x}}$. Here each $x_i$ is drawn from a standard multivariate Gaussian distribution $\mathcal{N}(0,\mathbf{I}_{2})$ and Let $w^* = (0.5,0.5)$, then each $y_i$ is drawn from the Bernoulli distribution with the probability $P(Y_i = 1|x_i, w^*) = \sigma(-x^T_iw^*)$. We also restrict hypothesis space as $\mathcal{W} = \{w: \|w\|_2 < 3\}$ where $W_{\ERM}$ falls in this area with high probability. Since the hypothesis is bounded and under the log-loss, then the learning problem will satisfy the central and witness condition \cite{van2015fast,Grunwald2020}. Therefore, it will satisfy the $(\eta,c)$-central condition. We will evaluate the generalization error and excess risk bounds in (\ref{thm:eta-c}). To this end, we need to estimate $\eta$, $c$ and mutual information $I(W_{\ERM},Z_i)$ efficiently, hence we repeatedly generate $W_{\ERM}$ and $Z_i$ and use the empirical density for estimation. Specifically, we vary the sample size $n$ from $200$ to $1600$ and for each $n$ we repeat the logistic regression algorithm 10000 times to generate a set of $W_\ERM$. For $\eta$ and $c$, we can empirically estimate the CGF and expected excess risk with the data sample and a set of ERM hypothesis. For the mutual information, we decompose $I(W;X,Y) =I(W;Y) + P(Y = 0)I(W;X|y = 0) + P(Y = 1)I(W;X|y = 1)$ by chain rule, and the first term can be approximated using the continuous-discrete estimator\cite{gao2017estimating} for mutual information and the rest terms are continuous-continuous ones\cite{moddemeijer1989estimation,kraskov2004estimating}. To demonstrate the usefulness of the results, we also compare the bounds with the true excess risk and true generalization error. The comparisons are shown in Figure~\ref{fig:logistic}.
\begin{figure}[H]
	\centering
	\subfloat[Generalization Error]{\includegraphics[width=0.35\textwidth]{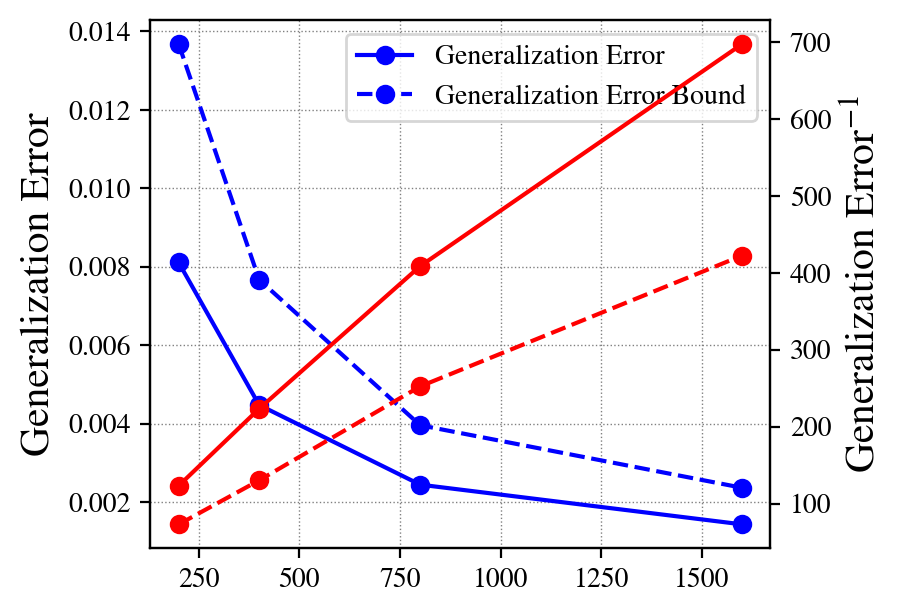}}	\quad
	\subfloat[Excess Risk]{\includegraphics[width=0.35\textwidth]{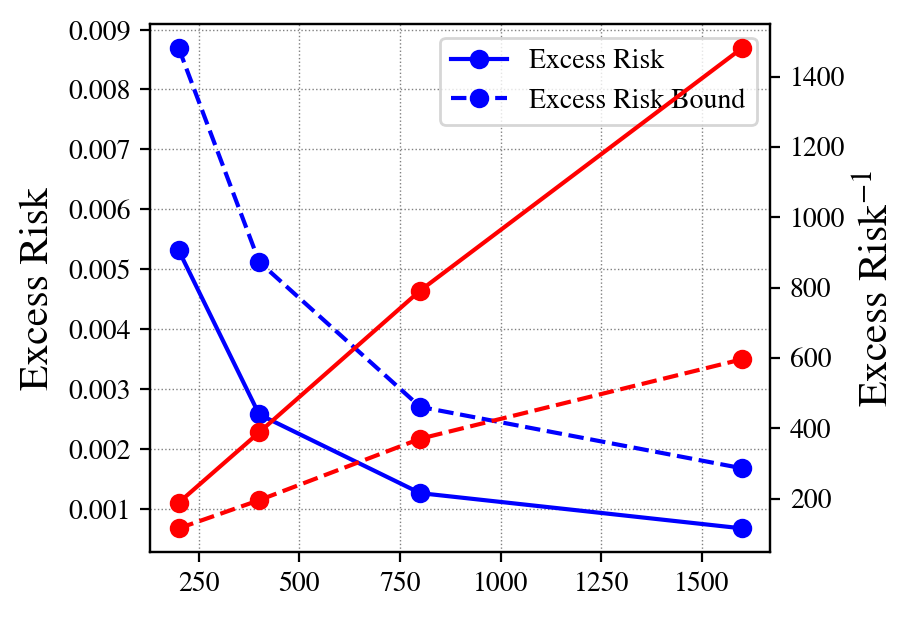}}
    \caption{We represent the true expected generalization error in (a) and true excess risk in (b) along with their bounds in Theorem~\ref{thm:eta-c}. Here we vary $n$ from 200 to 1600. We also plot their reciprocals to show the rate w.r.t. sample size $n$.   All results are derived by 10000 experimental repeats.}\label{fig:logistic}
\end{figure}
From the figure, we can see that both the generalization error and excess risk converge as $O(\frac{1}{n})$, and the bounds in Theorem~\ref{thm:eta-c} are tight, which capture the true behaviours with the same decay rate.

\subsection{Proof of~Theorem~\ref{thm:subgaussian}}
\begin{proof}
It is found that, due to $w^*$ is independent of $Z_i$, we have
\begin{equation*}
    \mathbb{E}_{W\mathcal{S}_n}[\mathcal{E}(W, \mathcal{S}_n)] = \mathbb{E}_{W\mathcal{S}_n}[\mathcal{R}(W) - \hat{\mathcal{R}}(W, \mathcal{S}_n)] = \mathbb{E}_{W \otimes \mathcal{S}_n}[\hat{\mathcal{R}}(W,\mathcal{S}_n)] - \mathbb{E}_{W\mathcal{S}_n}[\hat{\mathcal{R}}(W,\mathcal{S}_n)] .
\end{equation*}
Let the distribution $P_{WZ_i}$ denote the joint distribution induced by $P_{W\mathcal{S}_n}$ with the algorihtm $P_{W|\mathcal{S}_n}$. With the i.i.d. assumption, we can rewrite the generalization error by,
\begin{equation*}
    \mathbb{E}_{W\mathcal{S}_n}[\mathcal{E}(W,\mathcal{S}_n)] = \frac{1}{n}\sum_{i=1}^n \mathbb{E}_{P_W \otimes \mu}[r(W,Z_i)] - \mathbb{E}_{WZ_i}[r(W,Z_i)].
\end{equation*}
Using the KL-divergence property \cite{xu2017information,bu2020tightening} that
\begin{align*}
    \sqrt{ 2\sigma^2 D\left( P_{WZ_i} \| P_{W}\otimes P_{Z_i} \right)} \geq \mathbb{E}_{P_W \otimes \mu}[r(W,Z_i)] - \mathbb{E}_{WZ_i}[r(W,Z_i)]
\end{align*}
under the $\sigma$-subgaussian assumption under the distribution $P_{W} \otimes \mu$. Summing up every term concludes the proof.
\end{proof}
\subsection{Proof of~Theorem~\ref{thm:subgaussianv2}}
\begin{proof}
Using the Donsker-Varadhan representation of the KL divergence, we build on the following inequality for some $\eta > 0$,
\begin{align}
    \frac{I(W;Z_i)}{\eta} + \Esub{P_{WZ_i}}{r(W,Z_i)} &\geq  - \frac{1}{\eta} \log \mathbb{E}_{P_{W}\otimes \mu}[e^{-\eta(r(W,Z_i))}].
\end{align}
We will bound the R.H.S. using the following technique. For any $a$, we have,
\begin{align}
    \log \mathbb{E}_{P_{W}\otimes \mu}[e^{a\eta \mathbb{E}[r(W,Z_i)] -\eta(r(W,Z_i))}] &= \log \mathbb{E}_{P_{W}\otimes \mu}[e^{a\eta \mathbb{E}[r(W,Z_i)] - \eta(r(W,Z_i))}] \\
    &= \log \mathbb{E}_{P_{W}\otimes \mu}[e^{\eta (\mathbb{E}[r(W,Z_i)] - r(W,Z_i)) + (a-1)\eta \mathbb{E}[r(W,Z_i)]}] \\
    &\leq \frac{\sigma^2 \eta^2}{2} + (a-1) \eta \mathbb{E}[r(W,Z_i)]. 
\end{align}
By setting $0 < a_\eta = 1-  \frac{\eta\sigma^2}{2\mathbb{E}[r(W,Z_i)]} < 1$, we have,
\begin{align}
    0 < \eta < \frac{2\mathbb{E}[r(W;Z_i)]}{\sigma^2},
\end{align}
and
\begin{align}
    \log \mathbb{E}_{P_{W}\otimes \mu}[e^{-\eta(r(W,Z_i))}] &\leq -a_\eta \eta \mathbb{E}[r(W,Z_i)]. 
\end{align}
We then rewrite the inequality by,
\begin{align}
    -\frac{1}{\eta}\log \mathbb{E}_{P_{W}\otimes \mu}[e^{-\eta(r(W,Z_i))}] &\geq a_\eta \mathbb{E}[r(W,Z_i)], 
\end{align}
and we further have the following bound,
\begin{align}
    \frac{I(W;Z_i)}{\eta} + \mathbb{E}_{P_{WZ_i}}[r(W,Z_i)] &\geq a_\eta \mathbb{E}_{P_WP_{Z_i}}[r(W,Z_i)].
\end{align}
Hence,
\begin{align}
    \mathbb{E}_{P_WP_{Z_i}}[r(W,Z_i)] - \mathbb{E}_{P_{WZ_i}}[r(W,Z_i)]  \leq \frac{I(W;Z_i)}{\eta a_\eta} + \frac{1-a_\eta}{a_\eta}\mathbb{E}_{P_{WZ_i}}[r(W,Z_i)].
\end{align}
Summing every term for $Z_i$, we have,
\begin{align}
     \mathbb{E}_{W\mathcal{S}_n} \left[\mathcal{E}(W, \mathcal{S}_n)\right] \leq & \frac{1-a_\eta}{a_\eta} \mathbb{E}_{W\mathcal{S}_n}[\hat{\mathcal{R}(W,\mathcal{S}_n)}] + \frac{1}{n\eta a_\eta} \sum_{i=1}^{n}  I\left(W ; Z_{i}\right),
\end{align}
which completes the proof.
\end{proof}

\subsection{Proof of~Corollary~\ref{coro:berstein}}
\begin{proof}
We firstly present the expected Bernstein inequality which will be the key technical lemma for the fast rate bound.
\begin{lemma}[Expected Bernstein Inequality \cite{mhammedi2019pac,cesa2006prediction}] \label{lemma:exp_bern}
Let $U$ be a random variable bounded from below by $-b < 0 $ almost surely, and let $\kappa(x)=(e^x - x - 1) / x^{2} .$ For all $\eta >0$, we have
\begin{align*}
 \log \mathbb{E}_{U}\left[e^{\eta(\mathbb{E}[U]-U}) \right] \leq \eta^2 c_{\eta} \cdot \mathbb{E}[U^{2}], \quad \text { for all } c_{\eta} \geq  \kappa(\eta b).
\end{align*}
\end{lemma}
\begin{proof}
The proof of the lemma follows from~\cite{cesa2006prediction} and \cite{mhammedi2019pac}. Firstly we define $Y = - U$ which is upper bounded by $b$, then using the property that $\frac{e^Y - Y -1}{Y^2}$ in non-increasing for $Y \in \mathbb{R}$, then we define $Z = \eta Y$ such that,
\begin{align*}
   \frac{e^{Z} - Z -1}{Z^2} \leq \frac{e^{\eta b} - \eta b - 1}{\eta^2 b^2} = \kappa (\eta b).
\end{align*}
Rearranging the inequality, we then arrive at,
\begin{align*}
   e^{Z} - Z -1 \leq Z^2 \kappa (\eta b).
\end{align*}
Taking the expectation on both sides and using the fact that $\log (x+1) \leq x$ for any $x \in \mathbb{R}$, we have,
\begin{align*}
   \log\mathbb{E}[e^{Z}]  - \mathbb{E}[Z] \leq \mathbb{E}[Z^2] \kappa (\eta b).
\end{align*}
By substituting $Z = \eta Y$, we have
\begin{align*}
   \mathbb{E}[e^{\eta (Y - \mathbb{E}[Y])}] \leq e^{\eta^2 \mathbb{E}[Y^2] \kappa (\eta b)}.
\end{align*}
Define $c_\eta \geq \kappa(\eta b)$, it yields that
\begin{align*}
   \mathbb{E}[e^{\eta (Y - \mathbb{E}[Y])}] \leq e^{\eta^2 c_\eta \mathbb{E}[Y^2] }.
\end{align*}
By substituting $Y = -U$, we finally have,
\begin{align*}
   \mathbb{E}[e^{\eta (\mathbb{E}[U] - U)}] \leq e^{\eta^2 c_\eta \mathbb{E}[U^2] },
\end{align*}
which completes the proof.
\end{proof}
Using the Bernstein condition and we also assume that $r(w,z_i)$ is lower bounded by $-b$ almost surely, we have for all $0< \eta < \frac{1}{b}$ and all $c > \frac{e^{\eta b} - \eta b - 1}{\eta^2b^2} > 0$, the following inequality holds:
 \begin{align}
     e^{\eta(\Esub{P_{W}\otimes \mu}{r(W,Z_i)} - r(W,Z_i))} \leq e^{\eta^2 c \Esub{P_{W}\otimes \mu}{r^2(W,Z_i)}}. \label{eq:rsquare}
 \end{align}
 With Lemma~\ref{lemma:exp_bern}, we have that for some $\beta \in [0,1]$, any $c>0$ and $\eta < \frac{1}{2Bc}$, then for all $0 < \beta' \leq \beta$ we have
\begin{align}
    \eta^2c \Esub{P_W \otimes \mu}{r^2(w,Z_i)} \leq \left(\frac{1}{2} \wedge \beta'\right) \eta  \left({\mathbb{E}_{P_W\otimes \mu}}[r(w,Z_i)]\right)+(1-\beta') \cdot(2 B c \eta)^{\frac{1}{1-\beta'}}\eta.
\end{align}
Hence the equation~(\ref{eq:rsquare}) can be further bounded by
\begin{align}
    e^{\eta(\Esub{P_{W}\otimes \mu}{r(W,Z_i)} - r(W,Z_i))} \leq e^{\left(\frac{1}{2} \wedge \beta'\right) \eta  \left({\mathbb{E}_{P_W \otimes \mu}}[r(w,Z_i)]\right)+(1-\beta') \cdot(2 B c \eta)^{\frac{1}{1-\beta'}}\eta } 
    \label{bound:bernstein}
\end{align}
for $\eta' < \min(\frac{1}{2B(e-1)}, \frac{1}{b})$. Here we can choose $c$ to be $\max_{\eta} \frac{e^{\eta b} - \eta b - 1}{\eta^2b^2} = e-1$ since the function $\frac{e^x - x - 1}{x^2}$ is non-decreasing in $[0,1]$. Furthermore, if $\beta' = 1$, we can rewrite~(\ref{bound:bernstein}) as,
\begin{align}
    e^{\eta(\Esub{P_{W}\otimes \mu}{r(W,Z_i)} - r(W,Z_i))} \leq e^{\frac{1}{2} \eta  \left({\mathbb{E}_{P_W \otimes \mu}}[r(w,Z_i)]\right)} 
\end{align}
which completes the proof.
\end{proof}

\subsection{Proof of~Corollary~\ref{coro:central}}
\begin{proof}
We first present the following Lemma for bounding the excess risk using the cumulant generating function.
\begin{lemma}[Generalized from Lemma 13 in \cite{Grunwald2020}]\label{lemma:central}
Let $\bar{\eta}>0$. Assume that the expected $\eta$-strong central condition holds, and suppose further that the $(u, c)$-witness condition holds for $u>0$ and $c \in(0,1]$. Let $0 < \eta' < \eta$ and $c_{u}:=\frac{1}{c} \frac{\eta' u+1}{1-\frac{\eta'}{\eta}} > 1$, then the following inequality holds:
\begin{equation}
\mathbb{E}_{P_W \otimes \mu}\left[r(W,Z)\right]  \leq - \frac{c_{u}}{\eta'}  \log \mathbb{E}_{ P_W \otimes \mu}\left[e^{-\eta' r(W,Z)}\right] . 
\end{equation}

\end{lemma}
The proof of the above lemma is similar to the proof in Appendix C.1 (page 48) in \cite{Grunwald2020} by taking the expectation over the hypothesis distribution $P_W$, which is omitted here. Now with~Lemma~\ref{lemma:central}, we have that for any $0 < \eta' < \eta$,
\begin{align}
  \log \mathbb{E}_{ P_W \otimes \mu}\left[e^{-\eta' \left( r(W,Z) - \mathbb{E}_{P_W \otimes \mu}[r(W,Z)]  \right)}\right] &\leq -\frac{\eta'}{c_u}\mathbb{E}_{P_W \otimes \mu}\left[r(W,Z)\right] + \eta'\mathbb{E}_{P_W \otimes \mu}\left[r(W,Z)\right] \\
  &=  (1-\frac{1}{c_u}) \eta'  \mathbb{E}_{P_W \otimes \mu}\left[r(W,Z)\right].
\end{align}
Therefore, the central condition with the witness condition implies the expected $(\eta', \frac{c-\frac{c\eta'}{\eta}}{\eta' u +1})$-central condition for any $0 < \eta' < \eta$.
\end{proof}

\subsection{Proof of Theorem~\ref{thm:eta-c}}
\begin{proof}
Firstly we rewrite excess risk and empirical excess risk by:
    \begin{align}
        \mathcal{R}(w) &= \mathbb{E}_{Z\sim \mu}[\ell(w,Z)] - \mathbb{E}_{Z\sim \mu}[\ell(w^*,Z)] \nonumber \\
        &= \frac{1}{n}\sum_{i=1}^{n} \mathbb{E}_{Z_i\sim \mu}[\ell(w,Z_i)] - \mathbb{E}_{Z_i \sim \mu}[\ell(w^*,Z_i)] \nonumber \\
        &= \mathbb{E}_{\mathcal{S}_n}[\hat{\mathcal{R}}(w, \mathcal{S}_n)],
    \end{align}
    and 
    \begin{align}
        \hat{\mathcal{R}}(w, \mathcal{S}_n) &=  \hat{L}(w,\mathcal{S}_n) - \hat{L}(w^*,\mathcal{S}_n).
    \end{align}
Given any $\mathcal{S}_n$, the gap between the excess risk and empirical excess risk can be written as,
    \begin{align}
    \mathcal{R}(w) - \hat{\mathcal{R}}(w, \mathcal{S}_n) = \mathbb{E}_{\mathcal{S}_n}[\hat{\mathcal{R}}(w, \mathcal{S}_n)] - \hat{\mathcal{R}}(w, \mathcal{S}_n).
    \end{align}
We will bound the above quantity by taking the expectation w.r.t. $w$ learned from $\mathcal{S}_n$ by: 
    \begin{align}
        \mathbb{E}_{W\mathcal{S}_n}[\mathcal{E}(W)] &= \mathbb{E}_{W\mathcal{S}_n}[\mathcal{R}(w) - \hat{\mathcal{R}}(w, \mathcal{S}_n)] \\
        &=  \mathbb{E}_{P_W\otimes \mathcal{S}_n}[\hat{\mathcal{R}}(w, \mathcal{S}_n)] - \mathbb{E}_{W\mathcal{S}_n}[\hat{\mathcal{R}}(w, \mathcal{S}_n)] \\
        &= \frac{1}{n}\sum_{i=1}^n \mathbb{E}_{P_W \otimes \mu}[r(W,Z_i)] - \mathbb{E}_{WZ_i}[r(W,Z_i)].
    \end{align}
Recall that the variational representation of the KL divergence between two distributions $P$ and $Q$ defined over $\mathcal X$ is given as (see, e. g. \cite{boucheron_concentration_2013})
\begin{align*}
D(P||Q)=\sup_{f}\{\Esub{P}{f(X)}-\log\Esub{Q}{e^{f(x)}} \},
\end{align*}
where the supremum is taken over all measurable functions such that $\Esub{Q}{e^{f(x)}}$ exists. Under the expected $(\eta,c)$-central condition, for any $0< \eta' \leq \eta$, let $f(w,z_i) = -\eta' r(w,z_i)$, we have,
\begin{align}
    D(P_{WZ_i}\|P_{W}\otimes P_{Z_i}) &\geq \Esub{P_{WZ_i}}{-\eta' r(W,Z_i)} - \log \mathbb{E}_{P_{W}\otimes \mu}[e^{-\eta'(r(W,Z_i))}] \nonumber \\
    &= \Esub{P_{WZ_i}}{-\eta' r(W,Z_i)} - \log \mathbb{E}_{P_{W}\otimes \mu}[e^{-\eta'(r(W,Z_i) - \mathbb{E}_{P_{W}\otimes \mu}[r(W,Z_i)])  }] + \Esub{P_W \otimes \mu}{\eta' r(W,Z_i)} \nonumber \\
    &= \eta'\left(\Esub{P_W \otimes \mu}{r(W,Z_i)} - \Esub{P_{WZ_i}}{r(W,Z_i)}\right) - \log \mathbb{E}_{P_{W}\otimes \mu}[e^{\eta'( \mathbb{E}_{P_{W}\otimes \mu}[r(W,Z_i)]  - r(W,Z_i))  }]. \label{eq:MI_KL}
\end{align}
Next we will upper bound the second term $\log \mathbb{E}_{P_{W}\otimes \mu}[e^{\eta'( \mathbb{E}_{P_{W}\otimes \mu}[r(W,Z_i)]  - r(W,Z_i))  }]$ in R.H.S. using the expected $(\eta,c)$-central condition. From the $(\eta, c)$-central condition, we have,
\begin{align}
    \log \mathbb{E}_{P_{W}\otimes \mu}[e^{\eta( \mathbb{E}_{P_{W}\otimes \mu}[r(W,Z_i)]  - r(W,Z_i))  }] \leq (1-c)\eta \mathbb{E}_{P_{W}\otimes \mu}[r(W,Z_i)].
\end{align}
Since $\eta' \leq \eta$, Jensen's inequality yields:
\begin{align}
    \log \mathbb{E}_{P_{W}\otimes \mu}[e^{\eta'( \mathbb{E}_{P_{W}\otimes \mu}[r(W,Z_i)]  - r(W,Z_i))  }]   &=  \log \mathbb{E}_{P_{W}\otimes \mu}[e^{\frac{\eta'}{\eta}\eta( \mathbb{E}_{P_{W}\otimes \mu}[r(W,Z_i)]  - r(W,Z_i))  }] \\
    &\leq \log \left( \mathbb{E}_{P_{W}\otimes \mu}[e^{\eta( \mathbb{E}_{P_{W}\otimes \mu}[r(W,Z_i)]  - r(W,Z_i))  }] \right)^{\frac{\eta'}{\eta}} \\
    &\leq \frac{\eta'}{\eta} (1-c)\eta \mathbb{E}_{P_{W}\otimes \mu}[r(W,Z_i)] \\
    &= \eta'(1-c) \mathbb{E}_{P_{W}\otimes \mu}[r(W,Z_i)].
    \label{eq:bound-central}
\end{align}
Substitute (\ref{eq:bound-central}) into (\ref{eq:MI_KL}), we arrive at,
\begin{align}
    I(W;Z_i) \geq \eta' \left(\Esub{P_W \otimes \mu}{r(W,Z_i)} - \Esub{P_{WZ_i}}{r(W,Z_i)}\right) - (1-c)\eta' \mathbb{E}_{P_{W}\otimes \mu}[r(W,Z_i)].
\end{align}
Divide $\eta'$ on both side, we arrive at,
\begin{align}
    \frac{I(W;Z_i)}{\eta'} \geq\Esub{P_W \otimes \mu}{r(W,Z_i)} - \Esub{P_{WZ_i}}{r(W,Z_i)} - (1-c) \mathbb{E}_{P_{W}\otimes \mu}[r(W,Z_i)].
\end{align}
Rearrange the equation and yields,
\begin{align}
    c\Esub{P_W \otimes \mu}{r(W,Z_i)} \leq \Esub{P_{WZ_i}}{r(W,Z_i)} + \frac{I(W;Z_i)}{\eta'}.
\end{align}
Therefore,
\begin{align}
    \Esub{P_W \otimes \mu}{r(W,Z_i)} - \Esub{P_{WZ_i}}{r(W,Z_i)} \leq  (\frac{1}{c} - 1)\left({\mathbb{E}_{P_{WZ_i}}}[r(w,Z_i)]\right) + \frac{I(W;Z_i)}{c\eta'}.
\end{align}
Summing up every term for $Z_i$ and divide by $n$, we end up with,
\begin{align}
    \Esub{P_W \otimes P_{\mathcal{S}_n}}{\hat{\mathcal{R}}(W,\mathcal{S}_n)} - \Esub{P_{W\mathcal{S}_n}}{\hat{\mathcal{R}}(W,\mathcal{S}_n)} \leq & (\frac{1}{c} - 1) \left({\mathbb{E}_{P_{W\mathcal{S}_n}}}[\hat{\mathcal{R}}(W,\mathcal{S}_n)]\right) + \frac{1}{n}\sum_{i=1}^{n}\frac{I(W;Z_i)}{c\eta'}.
\end{align}
Finally we completes the proof by,
\begin{align}
    \mathbb{E}_{W\mathcal{S}_n}[\mathcal{E}(W)] \leq (\frac{1}{c} - 1) \left({\mathbb{E}_{P_{W\mathcal{S}_n}}}[\hat{\mathcal{R}}(W,\mathcal{S}_n)]\right) + \frac{1}{n}\sum_{i=1}^{n}\frac{I(W;Z_i)}{c\eta'}.
\end{align}

\end{proof}

\subsection{Proof of Lemma~\ref{lemma:rerm}}
\begin{proof}
We first define,
\begin{align*}
    \hat{{L}}_{\textup{reg}}(w,\mathcal{S}_n) := \hat{{L}}(w,\mathcal{S}_n) + \frac{\lambda}{n}g(w).
\end{align*}
Based on Theorem~\ref{thm:eta-c}, we can bound the excess risk for $W_{\sf{RERM}}$ by,
\begin{align*}
     \mathbb{E}_{W}[\mathcal{R}(W_{\sf{RERM}})] & \leq  \frac{1}{c} \mathbb{E}_{P_{W\mathcal{S}_n}}[\hat{\mathcal{R}}\left(W_{\sf{RERM}}, \mathcal{S}_{n} \right)]  + \frac{1}{c\eta' n} \sum_{i=1}^{n} I(W_{\sf{RERM}};Z_i) \\
     &= \frac{1}{c} \left( \mathbb{E}_{P_{W\mathcal{S}_n}}[\hat{L}\left(W_{\sf{RERM}}, \mathcal{S}_{n} \right) - \hat{L}\left(w^*, \mathcal{S}_{n} \right)] \right)   + \frac{1}{c\eta' n} \sum_{i=1}^{n} I(W_{\sf{RERM}};Z_i) \\
     &\overset{(a)}{\leq} \frac{1}{c} \left( \mathbb{E}_{P_{W\mathcal{S}_n}}[\hat{{L}}_{\textup{reg}} \left(W_{\sf{RERM}}, \mathcal{S}_{n} \right)] - \mathbb{E}_{P_{W\mathcal{S}_n}}[\hat{{L}}_{\textup{reg}}\left(w^*, \mathcal{S}_{n}\right)] \right)  + \frac{\lambda B}{cn}+ \frac{1}{c\eta' n} \sum_{i=1}^{n} I(W_{\sf{RERM}};Z_i) \\
     & =  \frac{1}{c} \mathbb{E}_{P_{W\mathcal{S}_n}}[\hat{\mathcal{R}}_{\textup{reg}}\left(W_{\sf{RERM}}, \mathcal{S}_{n} \right)]  + \frac{1}{c\eta' n} \sum_{i=1}^{n} I(W_{\sf{RERM}};Z_i) \\
     &\overset{(b)}{\leq}  \frac{\lambda B}{cn} + \frac{1}{c\eta' n} \sum_{i=1}^{n} I(W_{\sf{RERM}};Z_i).
 \end{align*}
where (a) follows since $|g(w^*) - g(W_{\sf{RERM}}))| \leq B$ the expected empirical risk is negative for $W_{\ERM}$ and (b) holds due to that $W_{\sf{RERM}}$ is the minimizer of the regularized loss. 
\end{proof}

\subsection{Proof of Theorem~\ref{lemma:intermediate}}
\begin{proof}
We will build upon~(\ref{eq:MI_KL}). With the $(v,c)$-central condition, for any $\epsilon \geq 0$ and any $ 0 < \eta' \leq v(\epsilon)$, the Jensen's inequality yields:
\begin{align}
    \log \mathbb{E}_{P_{W}\otimes \mu}[e^{\eta'( \mathbb{E}_{P_{W}\otimes \mu}[r(W,Z_i)]  - r(W,Z_i))  }]   &=  \log \mathbb{E}_{P_{W}\otimes \mu}[e^{\frac{\eta'}{v(\epsilon)}v(\epsilon)( \mathbb{E}_{P_{W}\otimes \mu}[r(W,Z_i)]  - r(W,Z_i))  }] \\
    &\leq \log \left( \mathbb{E}_{P_{W}\otimes \mu}[e^{v(\epsilon)( \mathbb{E}_{P_{W}\otimes \mu}[r(W,Z_i)]  - r(W,Z_i))  }] \right)^{\frac{\eta'}{v(\epsilon)}} \\
    &\leq \frac{\eta'}{v(\epsilon)} \left( (1-c)v(\epsilon) \mathbb{E}_{P_{W}\otimes \mu}[r(W,Z_i)] + v(\epsilon) \epsilon \right) \\
    &= \eta'(1-c) \mathbb{E}_{P_{W}\otimes \mu}[r(W,Z_i)] + \eta'\epsilon.
    \label{eq:bound-v-central}
\end{align}
Substitute (\ref{eq:bound-v-central}) into (\ref{eq:MI_KL}), we arrive at,
\begin{align}
    I(W;Z_i) \geq \eta' \left(\Esub{P_W \otimes \mu}{r(W,Z_i)} - \Esub{P_{WZ_i}}{r(W,Z_i)}\right) - (1-c)\eta' \mathbb{E}_{P_{W}\otimes \mu}[r(W,Z_i)] - \eta' \epsilon.
\end{align}
Divide $\eta'$ on both side, we arrive at,
\begin{align}
    \frac{I(W;Z_i)}{\eta'} \geq\Esub{P_W \otimes \mu}{r(W,Z_i)} - \Esub{P_{WZ_i}}{r(W,Z_i)} - (1-c) \mathbb{E}_{P_{W}\otimes \mu}[r(W,Z_i)] - \epsilon.
\end{align}
Rearrange the equation and yields,
\begin{align}
    c\Esub{P_W \otimes \mu}{r(W,Z_i)} \leq \Esub{P_{WZ_i}}{r(W,Z_i)} + \frac{I(W;Z_i)}{\eta'} +\epsilon.
\end{align}
Therefore,
\begin{align}
    \Esub{P_W \otimes \mu}{r(W,Z_i)} - \Esub{P_{WZ_i}}{r(W,Z_i)} \leq  (\frac{1}{c} - 1)\left({\mathbb{E}_{P_{WZ_i}}}[r(w,Z_i)]\right) + \frac{I(W;Z_i)}{c\eta'} + \frac{\epsilon}{c}.
\end{align}
Summing up every term for $Z_i$ and divide by $n$, we end up with,
\begin{align}
    \Esub{P_W \otimes P_{\mathcal{S}_n}}{\hat{\mathcal{R}}(W,\mathcal{S}_n)} - \Esub{P_{W\mathcal{S}_n}}{\hat{\mathcal{R}}(W,\mathcal{S}_n)} \leq & (\frac{1}{c} - 1) \left({\mathbb{E}_{P_{W\mathcal{S}_n}}}[\hat{\mathcal{R}}(W,\mathcal{S}_n)]\right) + \frac{1}{n}\sum_{i=1}^{n}\left( \frac{I(W;Z_i)}{c\eta'} + \frac{\epsilon}{c}\right).
\end{align}
Finally we arrive at the following inequality:
\begin{align}
    \mathbb{E}_{W\mathcal{S}_n}[\mathcal{E}(W)] \leq (\frac{1}{c} - 1) \left({\mathbb{E}_{P_{W\mathcal{S}_n}}}[\hat{\mathcal{R}}(W,\mathcal{S}_n)]\right) + \frac{1}{n}\sum_{i=1}^{n}\left(\frac{I(W;Z_i)}{c\eta'} + \frac{\epsilon}{c} \right).
\end{align}
In particular, if $v(\epsilon) = \epsilon^{1-\beta}$ for some $\beta \in [0,1]$, then by choosing $\eta' = v(\epsilon)$ and $\frac{I(W;Z_i)}{c\eta'} + \frac{\epsilon}{c}$ is optimized when $\epsilon = I(W;Z_i)^{\frac{1}{2-\beta}}$ and the bound becomes,
\begin{align}
    \mathbb{E}_{W\mathcal{S}_n}[\mathcal{E}(W)] \leq (\frac{1}{c} - 1) \left({\mathbb{E}_{P_{W\mathcal{S}_n}}}[\hat{\mathcal{R}}(W,\mathcal{S}_n)]\right) + \frac{2}{nc}\sum_{i=1}^{n} I(W;Z_i)^{\frac{1}{2-\beta}},
\end{align}
which completes the proof.
\end{proof}

\subsection{Calculation Details}\label{apd:cal}
In this section, we present the calculation details of the Gaussian mean estimation case. Let us consider the 1D-Gaussian mean estimation problem. Let $\ell(w,z_i) = (w-z_i)^2$, each sample is drawn from some Gaussian distribution, i.e., $Z_i \sim \mathcal{N}(\mu, \sigma_N^2)$. Then the ERM algorithm arrives at,
\begin{equation}
 W_{\ERM} = \frac{1}{n} \sum_{i=1}^{n} Z_i \sim \mathcal{N}(\mu, \frac{\sigma_N^2}{n}).
\end{equation}
It can be easily calculated that the optimal $w^*$ satisfies:
\begin{equation}
    w^* = \argmin \Esub{Z}{\ell(w,Z)} =\argmin \Esub{Z}{(w-Z)^2} = \mu.
\end{equation}
Also it can be calculated that the expected excess risk is,
\begin{align}
    \Esub{W}{\mathcal{R}(W_\ERM)} &= \mathbb{E}_{W \otimes Z}[\ell(W_\ERM,Z)] - \mathbb{E}_{Z}[\ell(w^*,Z)] \\
    &= \mathbb{E}_{W \otimes Z}[(W_\ERM - Z)^2]  - \mathbb{E}_{Z}[(\mu - Z)^2]\\
    &=  \mu^2 + \frac{\sigma_N^2}{n} + \mu^2 + \sigma_N^2 - 2\mu^2 -  \mu^2 - \mu^2 -\sigma_N^2 + 2\mu^2 \\
    &= \frac{\sigma_N^2}{n}.
\end{align}
The corresponding empirical excess risk is given by,
\begin{align}
    \Esub{W\mathcal{S}_n}{\hat{\mathcal{R}}(W_\ERM,\mathcal{S}_n)} &= \Esub{W_\ERM \mathcal{S}_n}{\hat L(W_\ERM,\mathcal{S}_n) - \hat{L}(w^*,\mathcal{S}_n)} \\
            &=  \E{\frac{1}{n}\sum_{i=1}^{n}(W-Z_i)^2 - \frac{1}{n}\sum_{i=1}^{n}(\mu - Z_i)^2} \\
            &=  \E{\frac{1}{n}\sum_{i=1}^{n}(W^2 - \mu^2) - \frac{2}{n}\sum_{i=1}^{n}WZ_i - \mu Z_i} \\
            &= \mu^2+ \frac{\sigma_N^2}{n} - \mu^2 - 2\mu^2 - \frac{2}{n} \sigma_N^2 + 2\mu^2 \\
            &= -\frac{\sigma_N^2}{n}.
\end{align}
Then it yields the expected generalization error as,
\begin{align}
    \Esub{W\mathcal{S}_n}{\mathcal{E}(W_\ERM, \mathcal{S}_n)} &= \Esub{W\mathcal{S}_n}{\mathcal{R}(W_\ERM)-\hat {\mathcal{R}}(W_\ERM,\mathcal{S}_n)} \\
    &= \frac{\sigma_N^2}{n} - (- \frac{\sigma_N^2}{n}) \\
    &= \frac{2\sigma_N^2}{n}.
\end{align}
The expected loss can be calculated as,
\begin{align}
    \mathbb{E}_{P_W \otimes \mu}[\ell(W_{\ERM},Z)] = \frac{n+1}{n}\sigma_N^2 := \sigma_W^2.
\end{align}
Let us verify the moment generating functions for the squared loss. Since $\ell(W_\ERM,Z)$ is $\sigma^2_W\chi^2_1$ distributed, 
\begin{align}
    \log \mathbb{E}_{P_W \otimes \mu}[e^{\eta (W_\ERM - Z)^2}] = -\frac{1}{2}\log(1- 2\sigma^2_W \eta).
\end{align}
Hence, 
\begin{align}
    \log \mathbb{E}_{P_W \otimes \mu}[e^{\eta \left((W_\ERM - Z)^2 - \mathbb{E}[(W_\ERM - Z)^2] \right)}] = -\frac{1}{2}\log(1- 2\sigma^2_W \eta) - \sigma^2_W\eta.
\end{align}
It is easily to prove that for any $x\leq 0$, 
\begin{align}
    -\frac{1}{2}\log(1-2x) - x \leq x^2,
\end{align}
which yields,
\begin{align}
    \log \mathbb{E}_{P_W \otimes \mu}[e^{\eta \left((W_\ERM - Z)^2 - \mathbb{E}[(W_\ERM - Z)^2] \right)}] = -\frac{1}{2}\log(1- 2\sigma^2_W \eta) - \sigma^2_W\eta \leq \sigma_W^4 \eta^2 
\end{align}
for any $\eta < 0$. Therefore, $\ell(W,Z)$ is $\sqrt{2\sigma^4_W}$-sub-Gaussian and we can only achieve the slow rate of $O(\sqrt{1/n})$.

Now we introduce $w^*$ in the sequel as a comparison. For a given $\eta$ and $w$, we calculate the moment generating function for the term $ r(w,Z) =(w - Z)^2 - (w^* - Z)^2$ as follows.
\begin{align*}
 \mathbb{E}_{\mu}[e^{\eta r(w,Z)}] &= \frac{1}{\sqrt{2\pi \sigma_N^2}} \int  e^{-\frac{(z-\mu)^2}{2\sigma_N^2}}  e^{\eta ((w - z)^2 - (w^* - z)^2)}  dz \\
    &= \frac{1}{\sqrt{2\pi \sigma_N^2}} \int  e^{-\frac{(z-\mu)^2}{2\sigma_N^2} + \eta ((w - z)^2 - (w^* - z)^2)}  dz \\
    &= \frac{1}{\sqrt{2\pi \sigma_N^2}} \int  \operatorname{exp}\{-\frac{z^2 - 2\mu z + 4\eta \sigma_N^2 (w-\mu)z + \mu^2 }{2\sigma_N^2} - \eta (\mu^2 - w^2)\}  dz \\
    &= \frac{1}{\sqrt{2\pi \sigma_N^2}} \int  \operatorname{exp}\{-\frac{(z - \mu + 2\eta \sigma_N^2(w-\mu))^2}{2\sigma_N^2}\}  dz \operatorname{exp}\left( -2\mu\eta (w - \mu) + 2\eta^2\sigma_N^2 (w-\mu)^2 - \eta (\mu^2 - w^2) \right) \\
    &= \operatorname{exp}\left( -2\mu\eta (w - \mu) + 2\eta^2\sigma_N^2 (w-\mu)^2 - \eta (\mu^2 - w^2) \right) \\
    &= \operatorname{exp}\left( (2\eta^2\sigma_N^2 + \eta)(w-\mu)^2 \right).
\end{align*}
Taking expectation over $w$ w.r.t. ERM solution, we have,
\begin{align*}
 \mathbb{E}_{P_W \otimes \mu}[e^{\eta r(W,Z)}] &= \Esub{W}{\operatorname{exp}\left( (2\eta^2\sigma_N^2 + \eta) (w-\mu)^2 \right)} \\
    &= \frac{1}{\sqrt{2\pi \frac{\sigma_N^2}{n}}} \int  \operatorname{exp}\left(-\frac{(w-\mu)^2}{2\sigma_N^2/n} + (2\eta^2\sigma_N^2 + \eta) (w-\mu)^2 \right)  dw \\
    &= \frac{1}{\sqrt{2\pi \frac{\sigma_N^2}{n}}} \int  \operatorname{exp}\left( -\frac{(w-\mu)^2}{2} (\frac{n}{\sigma_N^2} - (4\eta^2\sigma_N^2 + 2\eta)) \right)  dw \\
    &= \frac{\sqrt{\frac{1}{\frac{n}{\sigma_N^2} - (4\eta^2\sigma_N^2 + 2\eta)}  \frac{n}{\sigma_N^2}} }{\sqrt{2\pi \frac{1}{\frac{n}{\sigma_N^2} - (4\eta^2\sigma_N^2 + 2\eta)}}} \int  \operatorname{exp}\left( -\frac{(w-\mu)^2}{2} (\frac{n}{\sigma_N^2} - (4\eta^2\sigma_N^2 + 2\eta)) \right)  dw \\
    &= \sqrt{\frac{n}{ n- (4\eta^2\sigma_N^4 + 2 \eta\sigma_N^2)}}.
\end{align*}
Therefore for large $n$ and any $\eta \in \mathbb{R}$, we arrive at, 
\begin{align*}
    \log  \mathbb{E}_{P_W \otimes \mu}[e^{\eta (r(W,Z) - \mathbb{E}[r(W,Z)]}] &= \frac{1}{2} \log \frac{n}{ n- (4\eta^2\sigma_N^4 + 2 \eta\sigma_N^2)} - \frac{\sigma_N^2}{n} \\
    & \approx \frac{1}{2} \frac{4\eta^2\sigma_N^4 + 2 \eta\sigma_N^2}{n} - \frac{\sigma_N^2}{n}  \\
    & = \frac{\sigma_N^2 + 2\eta^2 \sigma_N^4}{n} - \frac{\sigma_N^2}{n}  \\
    & = \frac{2\eta^2 \sigma_N^4}{n}.
\end{align*}
The mutual information can be calculated as,
\begin{align*}
    I(W_\ERM;Z_i) &= h(W_\ERM) - h(W_\ERM|Z_i)\\
    &= \frac{1}{2}\log \frac{2\pi e \sigma_N^2}{n} - \frac{1}{2}\log\frac{2\pi e (n-1)\sigma_N^2 }{n^2} \\
    & = \frac{1}{2}\log \frac{n}{n-1} \\
    & \approx \frac{1}{2n}
\end{align*}
for large $n$. 

\begin{table}[H]
    \centering
    \begin{tabular}{c|c}
    \hline 
     Quantity    &  Values/Distribution \\
     \hline 
     $\mathcal{S}_n$   &  $\{Z_1,Z_2,\cdots,Z_n\}$  \\
     $Z_i$     &     $\mathcal{N}(\mu,\sigma_N^2)$ \\
      $\ell(w,z)$   &  $(w-z)^2$ \\
      $\hat{L}(w,\mathcal{S}_n)$  & $\frac{1}{n}\sum_{i=1}^{n}\ell(w,z_i)$  \\ 
      $L(w)$  &  $\Esub{Z}{\ell(w,Z)}$ \\
      $W_\ERM$  & $\mathcal{N}(\mu,\frac{\sigma_N^2}{n})$    \\
      $w^*$   & $\mu$    \\
      $r(w,z)$  & $(w-z)^2 - (w^* - z)^2$  \\
      $\mathcal{R}(w)$  & $L(w) - L(w^*)$   \\ 
      $\hat{\mathcal{R}}(w,\mathcal{S}_n)$  & $\hat{L}(w) - \hat{L}(w^*)$    \\
      $\mathcal{E}(w,\mathcal{S}_n)$  & $L(w) - \frac{1}{n}\sum_{i=1}^{n}\ell(w,z_i)$   \\
      $M_{Z}[r(w,Z)]$ & $-\frac{1}{\eta} \log \mathbb{E}_Z\left[e^{-\eta r(w,Z)}\right]$   \\
      $M_{P_W\otimes Z}[r(w,Z)]$ & $-\frac{1}{\eta} \log \mathbb{E}_{P_W\otimes Z}\left[e^{-\eta r(W,Z)}\right]$   \\
      \hline 
      \hline
      $\Esub{W\mathcal{S}_n}{\hat{\mathcal{R}}(W_\ERM,\mathcal{S}_n)}$   &  $-\frac{\sigma_N^2}{n}$   \\
      $\Esub{W}{\mathcal{R}(W_\ERM)}/\mathbb{E}_{P_W\otimes \mu}[r(W,Z)]$   &  $\frac{\sigma_N^2}{n}$   \\
      $\Esub{W\mathcal{S}_n}{\mathcal{E}(W_\ERM, \mathcal{S}_n)}$ & $\frac{2\sigma_N^2}{n}$   \\
      $\mathcal{R}(w)/\mathbb{E}_{Z}[r(w,Z)]$  &  $(w-\mu)^2$    \\
      $\mathbb{E}_{Z}[e^{\eta r(w,Z)}]$ & $\operatorname{exp}\left( (2\eta^2\sigma_N^2 +\eta )(w-\mu)^2 \right)$   \\
      $\mathbb{E}_{P_W\otimes Z}[e^{\eta r(W,Z)}]$ & $\sqrt{\frac{n}{ n- (4\eta^2\sigma_N^4 + 2 \eta\sigma_N^2)}}$   \\
        $\mathbb{E}_{Z}[e^{-\eta r(w,Z)}]$ & $\operatorname{exp}\left( (2\eta^2\sigma_N^2 - \eta )(w-\mu)^2 \right)$   \\
      $\mathbb{E}_{P_W\otimes Z}[e^{-\eta r(W,Z)}]$ & $\sqrt{\frac{n}{ n- (4\eta^2\sigma_N^4 - 2 \eta\sigma_N^2)}}$   \\
      $M_{Z}[r(w,Z)]$ & $(1-2\eta\sigma_N^2)(w-\mu)^2$    \\
      $M_{P_W\otimes Z}[r(w,Z)]$ & $(1-2\eta\sigma_N^2)\frac{\sigma_N^2}{n}$   \\
      $\mathbb{E}_{Z}[r(w,Z)^2]$   & $(w - \mu)^4 + 4(w -\mu)^2\sigma_N^2$    \\
      $\mathbb{E}_{P_W\otimes Z}[r(W,Z)^2]$ & $\frac{3\sigma_N^4}{n^2}+ \frac{4\sigma_N^4}{n}$  \\
      $I(W;Z_i)$   &    $\frac{1}{2}\log\frac{n}{n-1}$ \\
      \hline 
    \end{tabular}
    \caption{Summarized Quantities}
    \label{tab:my_label}
\end{table}
We then summarize all the quantities of interest in the Table~\ref{tab:my_label} for references. From the table we can check conclude that for most fast rate conditions such as Berstein's condition, central condition and subgaussian condition, the results will hold in expectation but this is not the case for any $w \in \mathcal{W}$. To see this, we will check whether the condition in succession.
\begin{itemize}
    \item When checking $\eta$-central condition, 
    \begin{itemize}
        \item For any $w$, 
        \begin{align*}
            \mathbb{E}_{Z}[e^{-\eta r(w,Z)}] = \operatorname{exp}\left( (2\eta^2\sigma_N^2 -\eta )(w-\mu)^2 \right) \leq 1, 
        \end{align*}
        then we require $0 < \eta \leq \frac{1}{2\sigma_N^2}$.
        \item For $W_\ERM$,
        \begin{align*}
        \mathbb{E}_{W\otimes Z}[e^{-\eta r(W,Z)}]  = \sqrt{\frac{n}{ n- (4\eta^2\sigma_N^4 - 2 \eta\sigma_N^2)}} \leq 1 .
        \end{align*}
        then we require $0 < \eta \leq \frac{1}{2\sigma_N^2}$.
    \end{itemize}
    \item When checking Bernstein's condition,
    \begin{itemize}
        \item For any $w\in \mathcal{W}$,
        \begin{align*}
          \mathbb{E}_{Z}[r(w,Z)^2]  &=  (w - \mu)^4 + 4(w -\mu)^2\sigma_N^2 \\
          & \leq B(\mathbb{E}_{ Z}[r(w,Z)])^{\beta} =  B(w-\mu)^{2\beta}.
        \end{align*}
        Apparently, this does not hold for all $w \in \mathbb{R}$ when $\beta \in [0,1]$.
        \item For $W_\ERM$,
        \begin{align*}
        \mathbb{E}_{W\otimes Z}[r(W_\ERM,Z)^2]  &= \frac{3\sigma_N^4}{n^2}+ \frac{4\sigma_N^4}{n} \\
        &\leq B(\mathbb{E}_{W\otimes Z}[r(W_\ERM,Z)])^{\beta} \\
        &=  B(\frac{\sigma_N^2}{n})^{\beta}.
        \end{align*}
        This holds for $\beta = 1$ and $B = 7\sigma_N^2$.
    \end{itemize}
    \item When checking witness condition,
    \begin{itemize}
    \item For any $w \in \mathcal{W}$, we require that,
    \begin{align*}
        \mathbb{E}_{Z}\left[\left(r(w,Z) \right) \cdot \mathbf{1}_{\left\{r(w,Z) \leq u \right\}}\right] &\geq c \mathbb{E}_{Z}\left[r(w,Z) \right] = c(w-\mu)^2.
    \end{align*}
    There does not exists finite $c$ and $u$ that satisfy the above inequality, so the witness condition does not hold for all $w \in \mathbb{W}$.
    \item For $W_\ERM$,
    \begin{align*}
        & \mathbb{E}_{W\otimes Z}\left[\left(r(W_\ERM,Z) \right) \cdot \mathbf{1}_{\left\{r(W_\ERM,Z) \leq u \right\}}\right]   \geq c \mathbb{E}_{W\otimes Z}\left[r(W,Z) \right] = \frac{c\sigma_N^2}{n}.
    \end{align*}
    In this case with high probability $r(W,Z)$ approaches zero and there exists $u$ and $c$ satisfying the above inequality.
    \end{itemize}
    \item When checking sub-Gaussian condition,
    \begin{itemize}
    \item For $W_\ERM$, when $0 < \eta \leq \frac{1}{2\sigma_N^2}$, we have,
    \begin{align*}
        & \log \mathbb{E}_{W\otimes Z}\left[e^{-\eta \left( r(W_\ERM,Z) - \mathbb{E}[r(W_\ERM,Z)]\right)} \right]  \sim \frac{2\eta^2\sigma_N^4}{n}.
    \end{align*}
    Then it satisfy with the $\sigma'^2$-sub-Gaussian condition that $\sigma'^2 = \frac{4\sigma_N^4}{n}$.
    \item For any $w$, 
     \begin{equation}
          \log \mathbb{E}_{Z}[e^{\eta r(w,Z)}] = 2\eta^2\sigma_N^2(w-\mu)^2.
     \end{equation}
     Since $w$ is unbounded, it does not satisfy the sub-Gaussian assumption for all $w\in\mathcal{W}$.
    \end{itemize}
\end{itemize}

\end{document}